\documentclass[11pt]{article}

\usepackage[utf8]{inputenc}
\usepackage[english]{babel}

\usepackage{hyperref}
\usepackage{xcolor}
\usepackage{amsmath,amssymb}
\usepackage{enumerate}
\usepackage{amsthm}
\usepackage[all,2cell]{xy}
\usepackage{mathrsfs}
\usepackage{mathtools}
\usepackage{tikz}

\usepackage{bm}

\usepackage{geometry}
\def\baselinestretch{1.1}
\newtheorem{thm}{Theorem}[section]
\newtheorem{dfn}[thm]{Definition}
\newtheorem{prop}[thm]{Proposition}

\newtheorem{exmpl}[thm]{Example}
\newtheorem{obs}[thm]{Remark}

\def\beq{\begin{equation}}
\def\eeq{\end{equation}}
\def\bea{\begin{eqnarray}}
\def\eea{\end{eqnarray}}
\def\beann{\begin{eqnarray*}}
\def\eeann{\end{eqnarray*}}
\def\ben{\begin{enumerate}}
\def\een{\end{enumerate}}
\def\bit{\begin{itemize}}
\def\eit{\end{itemize}}

\newcommand\restr[2]{{
  \left.\kern-\nulldelimiterspace 
  #1 
  \right|_{#2} 
}}

\newcommand*{\transp}[2][-3mu]{\ensuremath{\mskip1mu\prescript{\smash{\mathrm t\mkern#1}}{}{\mathstrut#2}}}%

\newcommand{\R}{\mathbb{R}}
\newcommand{\C}{\mathcal{C}}
\renewcommand{\d}{\mathrm{d}}

\renewcommand{\L}{\mathcal{L}}
\newcommand{\F}{\mathcal{F}}
\renewcommand{\H}{\mathcal{H}}

\newcommand{\df}{\Omega}
\newcommand{\Cinfty}{\mathscr{C}^\infty}

\newcommand{\T}{\mathrm{T}}

\newcommand{\Lie}{\mathscr{L}}

\newcommand{\W}{\mathcal{W}}
\newcommand{\X}{\mathfrak{X}}
\newcommand{\Reeb}{\mathcal{R}}

\newcommand{\parder}[2]{\frac{\partial #1}{\partial #2}}
\newcommand{\dparder}[2]{\dfrac{\partial #1}{\partial #2}}

\DeclareMathOperator{\rk}{rank}
\def\d{\mathrm{d}}

\let\ds\displaystyle

\title{\huge\sffamily
\vskip -10mm
Skinner--Rusk formalism for \texorpdfstring{$k$}--contact systems}

\author{\Large\sffamily 
Xavier Gr\`acia%
\thanks{{\bf e}-{\it mail}:
xavier.gracia@upc.edu. (ORCID: 0000-0003-1006-4086).}\:,
Xavier Rivas%
\thanks{{\bf e}-{\it mail}:
xavier.rivas@upc.edu. (ORCID: 0000-0002-4175-5157).}\;  
and   
Narciso Rom\'an-Roy%
\thanks{{\bf e}-{\it mail}:
narciso.roman@upc.edu. (ORCID: 0000-0003-3663-9861).} 
\\[1ex]
\normalsize\itshape\sffamily 
Department of Mathematics,
Universitat Polit\`ecnica de Catalunya,
Barcelona, Catalonia, Spain
}

\date{\sffamily 1 February 2022}

\begin{document}

\maketitle

\leavevmode
\vadjust{\kern -15mm}
\begin{abstract}
\noindent
In previous papers, a geometric framework has been developed 
to describe non-conservative field theories 
as a kind of modified Lagrangian and Hamiltonian field theories.
This approach is that of $k$-contact Hamiltonian systems,
which is based on the $k$-symplectic formulation of field theories
as well as on contact geometry.
In this work we present the Skinner--Rusk unified setting
for these kinds of theories, which encompasses both 
the Lagrangian and Hamiltonian formalisms into a single picture.
This unified framework is specially useful when dealing with singular systems,
since: (i) it incorporates in a natural way the second-order condition 
for the solutions of field equations,
(ii) it allows to implement the Lagrangian and Hamiltonian constraint algorithms
in a unique simple way, and (iii) it gives the Legendre transformation,
so that the Lagrangian and the Hamiltonian formalisms are obtained straightforwardly.
We apply this description to several interesting physical examples:
the damped vibrating string,
the telegrapher's equations,
and Maxwell's equations with dissipation terms.
\end{abstract}


\noindent\textbf{MSC\,2020:}
70S05;
35Q61, 35R01, 53C15, 53D10, 53Z05, 58A10, 70G45, 70H45.




\medskip
\setcounter{tocdepth}{2}
{
\def\baselinestretch{1}
\small
\def\addvspace#1{\vskip 1pt}
\parskip 0pt plus 0.1mm
\tableofcontents
}

\newpage

\section{Introduction}

The study of non-conservative or dissipative systems in physics
and other branches of knowledge has been a subject of renewed interest in recent years 
with the integration of contact geometry to this end
\cite{Banyaga2016,Bravetti2017,Bravetti2018,BCT-2017,CG-2019,CCM-2018,DeLeon2019,
DeLeon2016b,GGMRR-2019b,Geiges2008,Goto-2016,Kholodenko2013,Lainz2018,LIU2018,RMS-2017}.
In particular, a geometric framework which is based on the contact geometry
and the $k$-symplectic setting of field theories has been recently introduced 
to describe classical field theories with dissipation.
The notion of $k$-contact Hamiltonian system was introduced in
\cite{GGMRR-2019}
and was used to describe several PDE's of interest. 
This was later applied to Lagrangian field theory \cite{GGMRR-2020}.

Nevertheless, in order to deal with non-regular systems more efficiently, 
a mixed formalism combining in a single description the 
Lagrangian and the Hamiltonian formalisms was developed,
with a phase space described by velocity as well momentum coordinates.
This is the so-called {\sl Skinner--Rusk} or {\sl unified formalism} developed in \cite{SR-83}
(although a previous description in coordinates had been made in \cite{Ka-82}).
Over the years, this formalism has been generalized so that it can be applied 
to very different types of systems (time-dependent, vakonomic and nonholonomic,
 control, and higher-order mechanics and field theories)
\cite{BEMMR-2007,BEMMR-2008,art:Campos_DeLeon_Martin_Vankerschaver09,CMC-2002,
art:Colombo_Martin_Zuccalli10,CLMM-2002,LMM-2003,ELMMR-04,GR-2018,GM-05,art:Gracia_Pons_Roman91,
art:Prieto_Roman11,art:Prieto_Roman12_1,PR-2015,Rey2004,RRSV-2011,art:Vitagliano10}.

As we have pointed out, the Skinner--Rusk formalism is particularly 
interesting when dealing with singular systems,
because of its special features.
First, regardless of the regularity of the Lagrangian,
in the Skinner--Rusk formalism the theory is always singular
and the field equations are not consistent.
Nevertheless, the formalism incorporates in a natural way the second-order or holonomy condition 
for the solutions of the field equations, even in the case
of singular Lagrangians
(remember that, for singular Lagrangians, in the Lagrangian formalism, this property is not necessarily satisfied and must be imposed ``ad hoc'').
As the field equations are not consistent,
we must implement the constraint algorithm
which allows us to find the maximal constraint submanifold (if it exists)
where there are solutions to the field equations fulfilling 
the holonomy condition.
However, the constraint algorithm is implemented only once,
since the Lagrangian and Hamiltonian versions of the constraint algorithm,
as well as the corresponding solutions to the Euler--Lagrange and the  Hamiltonian equations (the {\sl Hamilton--De Donder--Weyl equations\/}),
are recovered straightforwardly from the Skinner--Rusk formalism,
using the Legendre map.
Furthermore, the Legendre map, itself, is obtained as a consequence
of the consistency conditions.

Recently, the Skinner--Rusk setting has been applied
to mechanical contact systems \cite{LGLMR-2021,LGMMR-2020}.
The aim of the present work is to describe the Skinner--Rusk formalism
for classical field theories with dissipation.
We start from the Lagrangian and Hamiltonian $k$-contact description for these kinds of systems introduced in \cite{GGMRR-2019,GGMRR-2020}, 
generalizing the unified formalisms 
previously developed for contact mechanics in \cite{LGMMR-2020}
and for the $k$-symplectic formulation of classical field theories in \cite{Rey2004}.

We use these results to explore several interesting physical applications. 
A first example is the vibrating string equation with damping.
The second example consists in adding a
damping term to the Lagrangian that describes the massive scalar field equation
(the Klein--Gordon equation), which allows us to obtain
an equation which is closely related to the
telegrapher's equation.
Finally, we consider the Lagrangian of electromagnetism with a dissipation term,
which leads to the equation of damped electromagnetic waves.

The organization of the paper is the following:
First, Section \ref{2} is a review on the foundations of the
$k$-contact formulation of field theories with dissipation,
in which we also include the guidelines of the constraint algorithm for the singular cases.
Section \ref{3} contains the main results of the work: the 
unified $k$-contact formalism is presented and
developed in detail.
Finally, in Section \ref{4}, we analyze the above mentioned examples.

All the manifolds are real, second countable and of class $\Cinfty$.
Manifolds and mappings are assumed to be smooth.
Sum over crossed repeated indices is understood.

\section{Hamiltonian and Lagrangian formalisms of \texorpdfstring{$k$}--contact systems}
\label{2}

In this section we review the Hamiltonian and Lagrangian formalisms for $k$-contact field theories (see \cite{GGMRR-2019,GGMRR-2020} for the details).
We also discuss the singular case, which is interesting for the development of the Skinner--Rusk formalism.

\subsection{\texorpdfstring{$k$}--contact structures}


A regular distribution on $M$ is a subset $D\subset \T M$ such that 
$D_x\subset \T_xM$ is a vector subspace, for every $x\in M$, 
that can be spanned by a family of vector fields and has locally constant rank. 
We denote by $D^\circ$ the annihilator of a distribution $D$, 
which is a regular codistribution, i.e., 
a subset $C\subset\T^\ast M$ such that 
$C_x\subset \T^\ast_xM$ is a vector subspace, for every $x\in M$.

Every nonvanishing 1-form $\eta\in\Omega^1(M)$ defines a codistribution of rank 1 denoted by $\langle \eta\rangle\subset\T^\ast M$. Notice that the annihilator $\langle \eta\rangle^\circ\subset \T M$ of $\langle \eta\rangle$ has corank 1 
and is the kernel of the vector bundle morphism
\begin{equation*}
	\begin{matrix}
		\hat\eta\colon & \T M & \to & M\times \R\\
		& v_p & \mapsto & (p, \eta_p(v_p))
	\end{matrix} \ .
\end{equation*}
With all this in mind, for every set of $k$ 1-forms $\eta^1,\dotsc,\eta^k\in\Omega^1(M)$, we define
\begin{align*}
	\mathcal{C}^{\rm C} &= \langle \eta^1,\dotsc,\eta^k \rangle\subset \T^\ast M\,,\\
	\mathcal{D}^{\rm C} &= \left(\mathcal{C}^{\rm C}\right)^\circ = \ker\widehat{\eta^1}\cap\dotsb\cap\ker\widehat{\eta^k}\subset \T M\,\\
	\mathcal{D}^{\rm R} &= \ker\widehat{\d\eta^1}\cap\dotsb\cap\ker\widehat{\d\eta^k}\subset\T M\,,\\
	\mathcal{C}^{\rm R} &= \left(\mathcal{D}^{\rm R}\right)^\circ\subset\T^\ast M\,.
\end{align*}

\begin{dfn}
\label{dfn-kcontact-manifold}
Let $M$ be a manifold with $\dim M=m$.
A \textbf{$k$-contact structure} on $M$ is a family of $k$ smooth 1-forms $\eta^1,\dotsc,\eta^k\in\Omega^1(M)$, $k<m$, such that
\begin{enumerate}[(i)]
	\item The distribution $\mathcal{D}^{\rm C}$ is regular and has corank $k$ ($\eta^1\wedge\dotsb\wedge\eta^k\neq 0$).
	\item The distribution $\mathcal{D}^{\rm R}$ is regular and has rank $k$.
	\item $\mathcal{D}^{\rm C}\cap\mathcal{D}^{\rm R} = \{0\}$, or equivalently, $\bigcap_{\alpha=1}^k\big( \ker\widehat{\eta^\alpha}\cap\ker\widehat{\d\eta^\alpha} \big) = \{0\}$.
\end{enumerate}
	We say that $\mathcal{C}^{\rm C}$ is the \textbf{contact codistribution}, $\mathcal{D}^{\rm C}$ is the \textbf{contact distribution}, $\mathcal{D}^{\rm R}$ is the \textbf{Reeb distribution} and $\mathcal{C}^{\rm R}$ is the \textbf{Reeb codistribution}. A manifold $M$ equipped with a $k$-contact structure is a \textbf{$k$-contact manifold}.
\end{dfn}

\begin{obs}\rm
Given conditions ({\it i }) and ({\it ii }), condition ({\it iii }) is equivalent to
	$$ (iii\:') \qquad TM = \mathcal{D}^{\rm C}\oplus\mathcal{D}^{\rm R}\,. $$
For $k=1$ we recover the notion of contact manifold.
\end{obs}

\begin{thm}
	Let $(M,\eta^\alpha)$ be a $k$-contact manifold. Then:
	\begin{enumerate}[(1)]
		\item There exists a family of $k$ vector fields $R_\alpha\in\mathfrak{X}(M)$, called \textbf{Reeb vector fields}, uniquely defined by the equations
		\begin{equation*}
			\begin{cases}
				i(R_\beta)\eta^\alpha = \delta_\beta^\alpha\,,\\
				i(R_\beta)\d\eta^\alpha = 0\,.
			\end{cases}
		\end{equation*}
		\item The Reeb distribution $\mathcal{D}^{\rm R}$ is involutive and, therefore, integrable. It is generated by the Reeb vector fields.
	\end{enumerate}
\end{thm}

On every $k$-contact manifold $(M,\eta^\alpha)$ there are coordinates $(x^I,s^\alpha)$, 
called {\it \textbf{adapted coordinates}}, such that
$$ 
R_\alpha = \parder{}{s^\alpha}\ ,\qquad \eta^\alpha = \d s^\alpha - f_I^\alpha(x)\d x^I\, , 
$$
where the functions $f_I^\alpha$ depend only on the coordinates $x^I$.

\begin{exmpl}\rm
	Given $k\geq1$, the manifold $M = (\oplus^k\T^\ast Q)\times \R^k$ equipped with natural coordinates $(q^i,p_i,s^\alpha)$ has a canonical $k$-contact structure defined by the differential 1-forms
	$$ \eta^\alpha = \d s^\alpha - \theta^\alpha\,, $$
	where $\theta^\alpha$ is the pull-back of the canonical 1-form of $\T^\ast Q$ with respect to the projection $(\oplus^k\T^\ast Q)\times \R^k\to\T^\ast Q$ to the $\alpha$-th direct summand. Their local expressions in the natural coordinates $(q^i,p_i,s^\alpha)$ are
	$$ \eta^\alpha = \d s^\alpha - p_i^\alpha\d q^i\,. $$
	Hence, $\d\eta^\alpha = \d q^i\wedge\d p_i^\alpha$ and the Reeb vector fields are
	$$R_\alpha = \parder{}{s^\alpha}\,. $$
\end{exmpl}

\begin{thm}[$k$-contact Darboux theorem]
	Consider a $k$-contact manifold $(M,\eta^\alpha)$ of dimension $n+kn+k$ endowed with an integrable subdistribution $\mathcal{V}\subset\mathcal{D}^{\rm C}$ with $\rk\mathcal{V} = nk$. 
	Around every point of $M$, there exists a local chart $(\mathcal{U};q^i,p_i^\alpha,s^\alpha)$, $1\leq\alpha\leq k$, $1\leq i\leq n$, such that
	$$ \restr{\eta^\alpha}{\mathcal{U}} = \d s^\alpha - p_i^\alpha\d q^i\,. $$
	In these coordinates,
$$ 
\restr{\mathcal{D}^{\rm R}}{\mathcal{U}} = 
\left\langle R_\alpha = 
\parder{}{s^\alpha}\right\rangle,
\quad 
\restr{\mathcal{V}}{\mathcal{U}} = \left\langle\parder{}{p_i^\alpha}\right\rangle\,. 
$$
	These coordinates are called \textbf{Darboux coordinates}.
\end{thm}

\begin{obs}{\rm
When some of the conditions stated in Definition 
\ref{dfn-kcontact-manifold} do not hold
we say that $\eta^1,\dotsc,\eta^k\in\Omega^1(M)$ is a 
{\sl \textbf{$k$-precontact structure}} and that $(M;\eta^1,\dotsc,\eta^k)$
is a {\sl \textbf{$k$-precontact manifold}}.
For this kind of manifolds,  Reeb vector fields are not uniquely determined.
(The case $k=1$ has been analyzed in \cite{DeLeon2019}, 
where the properties of these so-called
{\sl precontact structures} and {\sl precontact manifolds} are studied in deep).
}\end{obs}

\subsection{Hamiltonian formalism}

Let $M$ be a manifold with $\dim M=m$. 
A {\it \textbf{$k$-vector field}} in $M$ is a section of the projection
$\tau_M\colon\oplus^k\T M\to M$; that is, a map ${\bf Y}\colon M\to\oplus^k\T M$
such that $\tau_M\circ{\bf Y}={\rm Id}_M$.
A $k$-vector field is specified by means of a set of $k$ vector fields $\mathbf{Y} = (Y_1,\dotsc,Y_k)$,
where $Y_\alpha=\tau_M^\alpha\circ{\bf Y}$, where $\tau_M^\alpha\colon\oplus^k\T M\to\T M$ is the canonical projection on the $\alpha$ factor.
An {\it\textbf{integral section}} of a $k$-vector field ${\bf Y}=(Y_{1},\dots, Y_{k})$
is a map $\psi\colon D\subset\R^k \rightarrow M$, such that
$$
\T\psi\circ\frac{\partial}{\partial t^\alpha}=Y_\alpha \circ \psi \ ,
$$
where $t=(t^1,\ldots,t^k)$ are the canonical coordinates of~$\R^k$.
Equivalently, an integral section satisfies the equation
$$
\psi'={\bf Y} \circ \psi ,
$$
where $\psi'\colon D\subset\R^k\to\oplus^k\T M$ is the first prolongation of $\psi$ to $\oplus^k\T M$.
defined by
$$
\psi'(t) = \left( \psi(t),\T\psi\left( \parder{}{t^1}\bigg\vert_t\right),\dotsc,\T\psi\left( \parder{}{t^k}\bigg\vert_t \right) \right) = (\psi(t),\psi'_\alpha(t))\,.
$$
A $k$-vector field ${\bf Y}$ is {\it\textbf{integrable}} if
every point of~$M$ belongs to the image of an integral section
of~${\bf Y}$.
If $(x^i)$ are local coordinates in $M$ and $\ds Y_\alpha= Y_\alpha^i \frac{\partial}{\partial x^i}$,
then $\psi$ is an integral section of $\mathbf{Y}$ if, and only if,
$$
\frac{\partial \psi^i}{\partial t^\alpha} = Y_\alpha^i(\phi) \ .
$$


\begin{dfn}
	A \textbf{$k$-contact Hamiltonian system} is a family $(M,\eta^\alpha,H)$, where $(M,\eta^\alpha)$ is a $k$-contact manifold and $H\in\Cinfty(M)$ is called a \textbf{Hamiltonian function}.
\end{dfn}

\begin{obs}{\rm
If $(M,\eta^\alpha)$ is a $k$-precontact manifold, then
$(M,\eta^\alpha,\H)$ is said to be a {\sl \textbf{$k$-precontact Hamiltonian system}}.
}\end{obs}

Given a $k$-contact Hamiltonian system $(M,\eta^\alpha,H)$, the {\it\textbf{$k$-contact Hamilton--De Donder--Weyl equations}} for a map $\psi\colon D\subset\R^k\to M$ are
\begin{equation}\label{k-contact-sections}
	\begin{cases}
		i(\psi_\alpha')\d\eta^\alpha = \left(\d H - (\Lie_{R_\alpha}H)\eta^\alpha\right)\circ\psi\,,\\
		i(\psi_\alpha')\eta^\alpha = -H\circ\psi\,,
	\end{cases}
\end{equation}
The {\it \textbf{$k$-contact Hamilton--De Donder--Weyl equations}} for a $k$-vector field $\mathbf{Y} = (Y_1,\dotsc,Y_k)$ in $M$ are
\begin{equation}\label{k-contact-fields}
	\begin{cases}
		i(Y_\alpha)\d\eta^\alpha = \d H - (\Lie_{R_\alpha}H)\eta^\alpha\,,\\
		i(Y_\alpha)\eta^\alpha = -H\,.
	\end{cases}
\end{equation}
The solutions to these equations are called \textbf{Hamiltonian $k$-vector fields}. These equations are equivalent to
$$ \begin{cases}
	\Lie_{Y_\alpha}\eta^\alpha = -(\Lie_{R_\alpha}H)\eta^\alpha\,,\\
	i(Y_\alpha)\eta^\alpha = -H\,.
\end{cases} $$
Notice that these equations are always consistent. However, their solutions are neither unique, nor necessarily integrable.

Given an integrable $k$-vector field $\mathbf{Y} = (Y_1,\dotsc,Y_k)$ in $M$, every integral section $\psi\colon D\subset\R^k\to M$ of $\mathbf{Y}$ satisfies the $k$-contact equations for sections \eqref{k-contact-sections} if and only if the $k$-vector field $\mathbf{Y}$ satisfies the $k$-contact equations for fields \eqref{k-contact-fields}. It is important to point out that equations \eqref{k-contact-sections} and \eqref{k-contact-fields} are not totally equivalent, since a solution to \eqref{k-contact-sections} may not be an integral section of some integrable $k$-vector field in $M$ solution to \eqref{k-contact-fields}.

In Darboux coordinates, if $\psi = (q^i(t^\beta), p_i^\alpha(t^\beta), s^\alpha(t^\beta))$, then $\displaystyle\psi_\alpha' = \left( q^i, p_i^\alpha, s^\alpha, \parder{q^i}{t^\beta}, \parder{p_i^\alpha}{t^\beta}, \parder{s^\alpha}{t^\beta} \right)$, and equations \eqref{k-contact-sections} read
\begin{equation}\label{Hamilton-sections}
	\begin{dcases}
		\parder{q^i}{t^\alpha} = \parder{H}{p_i^\alpha}\circ\psi\,,\\
		\parder{p_i^\alpha}{t^\alpha} = -\left( \parder{H}{q^i} + p_i^\alpha\parder{H}{s^\alpha} \right)\circ\psi\,,\\
		\parder{s^\alpha}{t^\alpha} = \left( p_i^\alpha\parder{H}{p_i^\alpha} - H \right)\circ\psi\,.
	\end{dcases}
\end{equation}

Let $\mathbf{Y}=(Y_1,\dotsc,Y_k)$ be a $k$-vector field solution to \eqref{k-contact-fields} written in Darboux coordinates as
$$ Y_\alpha = (Y_\alpha)^\beta\parder{}{s^\beta} + (Y_\alpha)^i\parder{}{q^i} + (Y_\alpha)_i^\beta\parder{}{p_i^\beta}\,, $$
then,
\begin{equation*}
	\begin{dcases}
		(Y_\alpha)^i = \parder{H}{p_i^\alpha}\,,\\
		(Y_\alpha)_i^\alpha = -\left( \parder{H}{q^i} + p_i^\alpha\parder{H}{s^\alpha} \right)\,,\\
		(Y_\alpha)^\alpha = p_i^\alpha\parder{H}{p_i^\alpha} - H\,.
	\end{dcases}
\end{equation*}

\subsection{Lagrangian formalism}

Consider the bundle $\oplus^k\T Q\times\R^k$ with natural coordinates $(q^i,v^i_\alpha,s^\alpha)$. We have the canonical projections
$$ \tau_1\colon\oplus^k\T Q\times\R^k\to\oplus^k\T Q\quad,\quad \tau^k\colon\oplus^k\T Q\times\R^k\to \T Q\quad,\quad s^\alpha\colon\oplus^k\T Q\times\R^k\to\R\,. $$
We can extend the canonical estructures (the Liouville vector field and the canonical $k$-tangent structure) in $\oplus^k\T Q$ to $\oplus^k\T Q\times\R^k$, which have the local expressions
$$ \Delta = v^i_\alpha\parder{}{v^i_\alpha}\quad,\quad J^\alpha = \parder{}{v^i_\alpha}\otimes\d q^i\,. $$
\begin{dfn}
	Let $\mathbf{X} = (X_\alpha)$ be a $k$-vector field in $\oplus^k\T Q\times\R^k\to\R$. We say that $\mathbf{X}$ is a \textbf{second order partial differential equation} ({\sc sopde}) if $J^\alpha(X_\alpha) = \Delta$.
\end{dfn}
In coordinates, a \textsc{sopde} has the expression
$$ X_\alpha = v_\alpha^i\parder{}{q^i} + (X_\alpha)_\beta^i\parder{}{v_\beta^i} + (X_\alpha)^\beta\parder{}{s^\beta}\,. $$
\begin{dfn}
	Consider a section $\psi\colon\R^k\to Q\times\R^k$ of the projection $Q\times\R^k\to \R^k$, where $\psi = (\phi,s^\alpha)$ with $\phi\colon \R^k\to Q$. The \textbf{first prolongation} of $\psi$ to $\oplus^k\T Q\times\R^k$ is the map $\psi^{[1]}\colon \R^k\to\oplus^k\T Q\times\R^k$ given by $\psi^{[1]} = (\phi',s^\alpha)$. We say that the map $\psi^{[1]}$ is \textbf{holonomic}.
\end{dfn}


\begin{prop}
	A $k$-vector field $\mathbf{X}$ in $\oplus^k\T Q\times\R^k\to\R$ is a {\sc sopde} if and only if its integral sections are holonomic.
\end{prop}
\begin{dfn}
	\begin{enumerate}[(1)]
		\item A \textbf{Lagrangian function} is a function $L\in\Cinfty(\oplus^k\T Q\times\R^k)$.
		\item The \textbf{Lagrangian energy} associated to $L$ is the function $E_L = \Delta(L)-L\in\Cinfty(\oplus^k\T Q\times\R^k\to\R)$.
		\item The \textbf{Cartan forms} associated to $L$ are
		$$ \theta_L^\alpha = \transp{(J^\alpha)}\circ\d L\in\Omega^1(\oplus^k\T Q\times\R^k)\quad,\quad\omega_L^\alpha = -\d\theta_L^\alpha\in\Omega^2(\oplus^k\T Q\times\R^k)\,. $$
		\item The \textbf{contact forms} associated to $L$ are
		$$ \eta_L^\alpha = \d s^\alpha-\theta_L^\alpha\in\Omega^1(\oplus^k\T Q\times\R^k)\quad,\quad\d\eta_L^\alpha = \omega_L^\alpha\in\Omega^2(\oplus^k\T Q\times\R^k)\,. $$
		\item The couple $(\oplus^k\T Q\times\R^k,L)$ is a \textbf{$k$-contact Lagrangian system}.
	\end{enumerate}
\end{dfn}
The local expression of these elements in natural coordinates of $\oplus^k\T Q\times\R^k$ are
$$ E_L = v_\alpha^i\parder{L}{v_\alpha^i} - L\quad,\quad \eta_L^\alpha = \d s^\alpha - \parder{L}{v_\alpha^i}\d q^i\,. $$

\begin{dfn}
	Given a Lagrangian function $L\in\Cinfty(\oplus^k\T Q\times\R^k)$ we define its \textbf{Legendre map} as the fibre derivative of $L$, considered as a function on the vector bundle $\oplus^k\T Q\times\R^k\to Q\times\R^k$, that is, the map $\F L\colon \oplus^k\T Q\times\R^k\to\oplus^k\T^\ast Q\times\R^k$ given by
	$$ \F L(v_{1q},\dotsc,v_{kq},s^\alpha) = (\F L(\cdot,s^\alpha)(v_{1q},\dotsc,v_{kq}),s^\alpha)\,. $$
\end{dfn}
In natural coordinates, the Legendre map is locally given by
\begin{equation}\label{legendre-coords}
	\F L(q^i,v_\alpha^i,s^\alpha) = \left( q^i,\parder{L}{v_\alpha^i},s^\alpha \right)\,.
\end{equation}
Notice that the Cartan forms can also be defined using the Legendre map as
$$ \theta_\L^\alpha = \F L^\ast\theta^\alpha\quad,\quad \omega_\L^\alpha = \F L^\ast\omega^\alpha\,, $$
where $\omega^\alpha = -\d\theta^\alpha$ and $\theta^\alpha$ is the pull-back of the canonical 1-form of $\T^\ast Q$ with respect to the projection $(\oplus^k\T^\ast Q)\times \R^k\to\T^\ast Q$ to the $\alpha$-th direct summand.

\begin{prop}\label{prop-legendre}
	Consider a Lagrangian function $L$. Then, the following conditions are equivalent:
	\begin{enumerate}[(1)]
		\item The Legendre map $\F L$ is a local diffeomorphism.
		\item The fibre Hessian $\F^2L\colon\oplus^k\T Q\times\R^k\to(\oplus^k\T^\ast Q\times\R^k)\otimes(\oplus^k\T^\ast Q\times\R^k)$ of $L$ is everywhere nondegenerate, where the tensor product is of vector bundles over $Q\times\R^k$.
		\item The couple $(\oplus^k\T Q\times\R^k,\eta_L^\alpha)$ is a $k$-contact manifold.
	\end{enumerate}
\end{prop}
The conditions in the above proposition are equivalent to require the matrix $W = W_{ij}^{\alpha\beta} = \bigg( \dfrac{\partial^2L}{\partial v_\alpha^i\partial v_\beta^j} \bigg)$ to be everywhere nonsingular.

\begin{dfn}
	A Lagrangian function $L$ is said to be \textbf{regular} if it satisfies the equivalent conditions in Proposition \ref{prop-legendre}. Otherwise, it is said to be a \textbf{singular} Lagrangian. If the Legendre map $\F L$ is a global diffeomorphism, $L$ is said to be \textbf{hyperregular}.
\end{dfn}

\begin{prop}
	Consider a regular $k$-contact Lagrangian system $(\oplus^k\T Q\times\R^k,L)$. There exists a unique family $(R_L)_\alpha\in\X(\oplus^k\T Q\times\R^k)$ of $k$ vector fields, called \textbf{Reeb vector fields}, satisfying the system of equations
	$$ \begin{cases}
		i( (R_L)_\alpha )\d\eta^\beta_L = 0\,,\\
		i( (R_L)_\alpha )\eta^\beta_L = \delta_\alpha^\beta\,.
	\end{cases} $$
\end{prop}
In natural coordinates, the Reeb vector fields have the local expressions
$$ (R_L)_\alpha = \parder{}{s^\alpha} - W_{\gamma\beta}^{ji}\frac{\partial^2L}{\partial s^\alpha\partial v_\gamma^j}\parder{}{v_\beta^i}\,, $$
where $W_{\alpha\beta}^{ij}$ is the inverse of the Hessian matrix $W_{ij}^{\alpha\beta}$, namely
$$ W_{\alpha\beta}^{ij}\frac{\partial^2L}{\partial v_\beta^j\partial v_\gamma^k} = \delta_k^i\delta_\alpha^\gamma\,. $$

Taking into account the previous results, it is clear that every regular
(resp., singular) Lagrangian $L\colon\oplus^k\T Q\times\R^k\to\R$ 
has associated the $k$-contact Hamiltonian system
(resp., $k$-precontact Hamiltonian system) $(\oplus^k\T Q\times\R^k, \eta^\alpha_L,E_L)$.
\begin{dfn}
	Consider a $k$-contact Lagrangian system $(\oplus^k\T Q\times\R^k, L)$.\\
	The \textbf{$k$-contact Euler--Lagrange equations} for a holonomic map $\psi\colon\R^k\to\oplus^k\T Q\times\R^k$ are
	\begin{equation}\label{k-contact-sections-lag}
		\begin{cases}
			i(\psi_\alpha')\d\eta_L^\alpha = \left( \d E_L - (\Lie_{(R_L)_\alpha}E_L)\eta_L^\alpha \right)\circ\psi\,,\\
			i(\psi_\alpha')\eta_L^\alpha = -E_L\circ\psi\,.
		\end{cases}
	\end{equation}
	The \textbf{$k$-contact Lagrangian equations} for a $k$-vector field $\mathbf{X} = (X_\alpha)$ in $\oplus^k\T Q\times\R^k$ are
	\begin{equation}\label{k-contact-fields-lag}
		\begin{cases}
			i(X_\alpha)\d\eta_L^\alpha = \d E_L - (\Lie_{(R_L)_\alpha}E_L)\eta_L^\alpha\,,\\
			i(X_\alpha)\eta_L^\alpha = -E_L\,.
		\end{cases}
	\end{equation}
	A $k$-vector fields solution to equations \eqref{k-contact-fields-lag} is called a \textbf{Lagrangian $k$-vector field}.
\end{dfn}
\begin{prop}
	Given a $k$-contact regular Lagrangian system $(\oplus^k\T Q\times\R^k, L)$, the system of equations \eqref{k-contact-fields-lag} is consistent. For $k>1$, the solutions are not unique.
\end{prop}

In canonical coordinates, equations \eqref{k-contact-sections-lag} read
\begin{equation}\label{euler-lagrange-integral-maps}
	\parder{}{t^\alpha}\parder{L}{v_\alpha^i}\circ\psi = \left( \parder{L}{q^i} + \parder{L}{s^\alpha}\parder{L}{v_\alpha^i} \right)\circ\psi\quad,\quad \parder{(s^\alpha\circ\psi)}{t^\alpha} = L\circ\psi\,.
\end{equation}
For a $k$-vector field $\mathbf{X} = (X_\alpha)$ with local expression
$$ X_\alpha = (X_\alpha)^i\parder{}{q^i} + (X_\alpha)_\beta^i\parder{}{v_\beta^i} + (X_\alpha)^\beta\parder{}{s^\beta}\,, $$
the $k$-contact Lagrangian equations \eqref{k-contact-fields-lag} read
\begin{align}
	0 &= \left( (X_\alpha)^j - v_\alpha^j \right) \frac{\partial^2L}{\partial v_\alpha^j\partial s^\beta}\,\label{lag-1}\\
	0 &= \left( (X_\alpha)^j - v_\alpha^j \right) \frac{\partial^2L}{\partial v_\beta^i\partial v_\alpha^j}\,\label{lag-2}\\
	0 &= \left( (X_\alpha)^j - v_\alpha^j \right) \frac{\partial^2L}{\partial q^i\partial v_\alpha^j} + \parder{L}{q^i} - \frac{\partial^2L}{\partial s^\beta\partial v_\alpha^i}(X_\alpha)^\beta - \frac{\partial^2L}{\partial q^j\partial v_\alpha^i}(X_\alpha)^j - \frac{\partial^2L}{\partial v^j_\beta\partial v_\alpha^i}(X_\alpha)^j_\beta + \parder{L}{s^\alpha}\parder{L}{v_\alpha^i}\,,\label{lag-3}\\
	0 &= L + \parder{L}{v_\alpha^i}\left( (X_\alpha)^j - v_\alpha^j \right) - (X_\alpha)^\alpha\,.\label{lag-4}
\end{align}
If the Lagrangian $L$ is regular, equations \eqref{lag-2} lead to the condition $v_\alpha^i = (X_\alpha)^i$, which are the \textsc{sopde} conditions for $\mathbf{X}$.
In this case, \eqref{lag-1} holds identically and, equations \eqref{lag-3} and \eqref{lag-4} give
\begin{align}
	(X_\alpha)^\alpha &= L\,,\label{Lagrange-1}\\
	-\parder{L}{q^i} + \frac{\partial^2L}{\partial s^\beta\partial v_\alpha^i}(X_\alpha)^\beta + \frac{\partial^2L}{\partial q^j\partial v_\alpha^i}v^j_\alpha + \frac{\partial^2L}{\partial v^j_\beta\partial v_\alpha^i}(X_\alpha)^j_\beta &= \parder{L}{s^\alpha}\parder{L}{v_\alpha^i}\label{Lagrange-2}\,.
\end{align}
It is important to point out that if the \textsc{sopde} $\mathbf{X}$ is integrable, these last equations are the Euler--Lagrange equations \eqref{euler-lagrange-integral-maps} for its integral maps.
\begin{prop}
	Given a regular Lagrangian $L$, the corresponding Lagrangian $k$-vector fields $\mathbf{X}$ are {\sc sopde}s. If, in addition, $\mathbf{X}$ is integrable, its integral sections are solutions to the $k$-contact Euler--Lagrange field equations \eqref{k-contact-sections-lag}.

	This {\sc sopde} $\mathbf{X}$ is called the \textbf{Euler--Lagrange $k$-vector field} associated to the Lagrangian~$L$.
\end{prop}

Notice that in the case $k=1$ we recover the Lagrangian formalism for contact systems \cite{GGMRR-2019b}.

\subsection{The singular case: 
\texorpdfstring{$k$}--precontact Lagrangian and Hamiltonian systems}
\label{singular}

For singular Lagrangians most of the results and properties stated in the above sections do not hold.

In this case, for the Lagrangian formalism, $(\oplus^k\T Q\times\R^k, \eta^\alpha_L)$
is not a $k$-contact manifold, but a $k$-precontact one, and hence the Reeb vector fields
are not uniquely defined.
Nevertheless, the Euler--Lagrange and the Lagrangian equations \eqref{k-contact-sections-lag}
and \eqref{k-contact-fields-lag} for the system $(\oplus^k\T Q\times\R^k,\eta^\alpha_L,E_L)$
are independent on the Reeb vector fields used 
(as it is proved in \cite{DeLeon2019} for the case $k=1$).
In any case, solutions to the Lagrangian equations
are not necessarily {\sc sopde} and this is a condition that 
must be added to the Lagrangian equation \eqref{k-contact-fields-lag}.
In addition, the field equations are not necessarily consistent everywhere on $\oplus^k\T Q\times\R^k$ 
and we must implement a {\sl constraint algorithm} to find 
a submanifold $S_f\hookrightarrow\oplus^k\T Q\times\R^k$
(if it exists) where there are {\sc sopde} $k$-vector fields in $\oplus^k\T Q\times\R^k$,
tangent to $S_f$, which are solutions to the equations \eqref{k-contact-fields-lag} on $S_f$.

In order to state the Hamiltonian formalism for the singular case, 
we need to assume some minimal regularity conditions. So, we define:

\begin{dfn}
	A singular Lagrangian $L$ is said to be \textbf{almost-regular} if
	\begin{enumerate}[(i)]
		\item The image of the Legendre map $\mathcal{P} = \F L(\oplus^k\T Q\times\R^k)\subseteq \oplus^k\T^\ast Q\times\R^k$ is a closed submanifold.
		\item $\F L$ is a submersion on ${\cal P}$.
		\item For every $p\in\mathcal{P}$, the fibre $\F L^{-1}(p)\subseteq\oplus^k\T Q\times\R^k$ is a connected submanifold. 
	\end{enumerate}
\end{dfn}

Then, if $j_{\cal P}\colon{\cal P}\hookrightarrow\oplus^k\T Q\times\R^k$
is the natural embedding and $\eta_{\cal P}=j_{\cal P}^*\eta^\alpha\in\df^1({\cal P})$,
we have that $({\cal P},\eta^\alpha_{\cal P})$
is, in general, a $k$-precontact manifold.
Furthermore, the function $E_L$ is  ${\cal F}L$-projectable and
there is a unique $H_{\cal P}\in\Cinfty({\cal P})$
such that $E_L={\cal F}L_o^*\,H_{\cal P}$,
where ${\cal F}L_o\colon\oplus^k\T Q\times\R^k\to{\cal P}$
is defined by ${\cal F}L=j_{\cal P}\circ{\cal F}L_o$.
Therefore, on the submanifold ${\cal P}$, 
there is a Hamiltonian formalism associated with the Lagrangian system,
and the $k$-contact Hamilton–De Donder–Weyl equations for a $k$-vector field
${\bf Y}=(Y_\alpha)$ in ${\cal P}$ are
\begin{equation}\label{HeqsP}
	\begin{cases}
i(Y_\alpha)\d\eta^\alpha_{\cal P} = \d H_{\cal P}-(\Lie_{R_\alpha}H_{\cal P})\eta^\alpha\ ,
\\
i(Y_\alpha)\eta^\alpha_{\cal P} = -H_{\cal P}\ .
	\end{cases}
\end{equation}
As in the Lagrangian formalism, these equations are not necessarily consistent everywhere on ${\cal P}$ 
and the constraint algorithm should also be implemented to find 
a submanifold $P_f\hookrightarrow{\cal P}$
(if it exists) where there are $k$-vector fields
tangent to $P_f$, which are solutions to the above equations \eqref{HeqsP} on $P_f$.

As a final remark, next we explain the guidelines of the constraint algorithm.
Consider a generic $k$-precontact Hamiltonian system $(M,\eta^\alpha,H)$
and its $k$-contact Hamiltonian field equations \eqref{k-contact-fields}

\begin{itemize}
\item
First we find the {\sl consistency conditions}:
Let $M_1$ be the subset of $M$ made of the points of $M$ where a solution to \eqref{k-contact-fields} exist, that is,
$$ 
M_1 = \{{\rm p}\in M\ \vert\ \exists(Y_1,\dotsc,Y_k)\in\oplus^k\T_qM \mbox{\rm\ solution to \eqref{k-contact-fields} at every ${\rm p}$}\}\ . 
$$
Assuming that $M_1\hookrightarrow M$ is a submanifold, there exists a section of the canonical projection $\tau_M\colon \oplus^k\T M\to M$ defined on $M_1$ which is a solution to \eqref{k-contact-fields}, but which may not be a $k$-vector field on $M_1$.
\item
Then we apply the {\sl tangency conditions}:
we define a new subset $M_2\subset M_1$ as
$$ 
M_2 = \{{\rm p}\in M_1\ \vert\ \exists(Y_1,\dotsc,Y_k)\in\oplus^k\T_qM_1\mbox{\ solution to \eqref{k-contact-fields} at every ${\rm p}$}\}\,. $$
Assuming that $M_2\hookrightarrow M_1$, then there is a section of the projection 
$\tau_{M_1}\colon \oplus^k\T M_1\to M_1$ defined on $M_2$ solution to equations 
\eqref{k-contact-fields} which does not define in general a $k$-vector field on $M_2$.

Taking a basis of independent constraint functions $\{\zeta^I\}$ locally defining $M_1$,
the constraints defining $M_2$ are obtained from
$$
(\Lie_{Y_{\alpha}}\zeta^I)\vert_{M_1}=0 \ .
$$
\item
Iterating this procedure we can obtain a sequence of constraint submanifolds
$$
\dotsb\hookrightarrow M_i\hookrightarrow\dotsb\hookrightarrow M_2\hookrightarrow M_1\hookrightarrow M\,.
$$
If this procedure stabilizes, that is, there exists a natural number $f\in\mathbb{N}$ such that 
$M_{f+1} = M_f$ and $\dim M_f>0$, we say that $M_f$ is the \textsl{final constraint submanifold}, 
where we can find solutions to equations \eqref{k-contact-fields}. 
Notice that the $k$-vector field solution may not be unique and, in general, they are not integrable.
\end{itemize}

\section{Skinner--Rusk unified formalism}
\label{3}

\subsection{Extended Pontryagin bundle: \texorpdfstring{$k$}--precontact canonical structure}

Consider a $k$-contact field theory with configuration space $Q\times\R^k$, where $\dim Q=n$, with coordinates $(q^i,s^\alpha)$. Now consider the bundles $\oplus^k\T Q\times\R^k$ and $\oplus^k\T^\ast Q\times\R^k$ equipped with natural coordinates $(q^i,v^i_\alpha, s^\alpha)$ and $(q^i, p_i^\alpha, s^\alpha)$ respectively.
We have the canonical projections
\begin{align*}
	\tau_1 &\colon \oplus^k\T Q\times\R^k \to \oplus^k\T Q & \tau_0&\colon \oplus^k\T Q\times\R^k \to Q\times\R^k\\
	\pi_1 &\colon \oplus^k\T^\ast Q\times\R^k \to \oplus^k\T^\ast Q & \pi_0&\colon \oplus^k\T^\ast Q\times\R^k \to Q\times\R^k \ .
\end{align*}
 We denote by $\d s^\alpha$ the volume form of $\R$ and its pull-backs to all the manifolds by the corresponding canonical projections. Consider the canonical forms $\theta_0\in\Omega^1(\T^\ast Q)$ and $\omega_0\in\Omega^2(\T^\ast Q)$ with local expressions $\theta_0 = p_i\d q^i$, $\omega_0 = \d q^i\wedge\d p_i$ in $\T^\ast Q$. We denote by $\theta^\alpha$ and $\omega^\alpha$ their pull-backs to $\oplus^k\T^\ast Q$ and $\oplus^k\T^\ast Q\times\R^k$, which have local expressions
$$ \theta^\alpha = p_i^\alpha\d q^i\quad,\quad\omega^\alpha = \d q^i\wedge\d p_i^\alpha\,. $$

\begin{dfn}
The \textbf{extended unified bundle} or \textbf{extended Pontryagin bundle} is
	$$ \W = \oplus^k\T Q\times_Q\oplus^k\T^\ast Q\times\R^k\,, $$
and it is endowed with the canonical projections
\begin{align*}
	\rho_1&\colon \W\to\oplus^k\T Q\times\R^k & 
	\rho_2&\colon \W\to\oplus^k\T^\ast Q\times\R^k \\
	\rho_0&\colon \W\to Q\times\R^k & 
	s^\alpha&\colon\W\to\R \ .
\end{align*}
\end{dfn}

The extended unified bundle has natural coordinates $(q^i,v^i_\alpha,p_i^\alpha,s^\alpha)$. 
We summarize all these manifolds and projections in the following diagram:
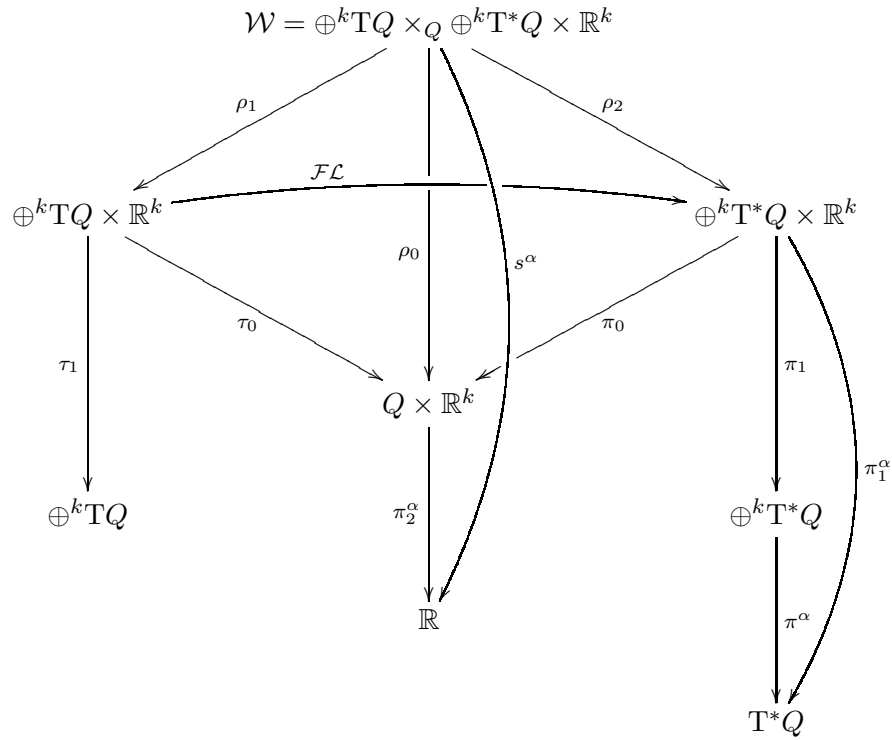
\begin{equation*}
\xymatrix{
	& \W = \oplus^k\T Q\times_Q\oplus^k\T^\ast Q\times\R^k \ar[ddr]^{\rho_2} \ar[ddl]_{\rho_1} \ar[dddd]_(0.6){\rho_0}|(0.415)\hole \ar@/^2.5pc/[dddddd]^(.4){s^\alpha} & \\ \\
    \oplus^k\T Q\times\R^k \ar[ddr]_{\tau_0} \ar[ddd]_{\tau_1} \ar@/^1pc/[rr]|(0.59)\hole^(0.35){\F\L} & & \oplus^k\T^\ast Q\times\R^k \ar[ddl]^{\pi_0}|(0.77)\hole \ar[ddd]^{\pi_1} \ar@/^2.5pc/[ddddd]^{\pi_1^\alpha} \\ \\
     & Q\times\R^k \ar[dd]_{\pi_2^\alpha} & \\
    \oplus^k\T Q &  & \oplus^k\T^\ast Q \ar[dd]^{\pi^\alpha} \\
    & \R & \\
    & & \T^\ast Q
}
\end{equation*}

\begin{dfn}
	Let $\psi\colon \R^k\to \W$ be a smooth map. We say that $\psi$ is \textbf{holonomic} if $\rho_1\circ\psi\colon\R^k\to\oplus^k\T Q\times\R^k$ is holonomic. A $k$-vector field $\mathbf{Z}\in\X^k(\W)$ is a \textbf{second order partial differential equation} ({\sc sopde} for short) if its integral sections are holonomic in $\W$.
\end{dfn}

In coordinates, a holonomic map $\psi\colon \R^k\to \W$ is expressed as
$$ \psi = \left( q^i(t), \parder{q^i}{t^\alpha}(t), p_{i}^\alpha(t),s^\alpha(t) \right)\,. $$ 
A $k$-vector field $\mathbf{Z} = (Z_1,\dotsc,Z_k)$ in $\W$ is a \textsc{sopde} if it has the following expression in natural coordinates:
$$ Z_\alpha = v^i_\alpha\parder{}{q^i} + (Z_\alpha)^i_\beta\parder{}{v^i_\beta} + (Z_\alpha)^\beta_i\parder{}{p_i^\beta} + (Z_\alpha)^\beta\parder{}{s^\beta}\,. $$

The extended unified bundle $\W$ is endowed with the following canonical structures:
\begin{dfn}
	\begin{enumerate}[(1)]
		\item The \textbf{coupling function} in $\W$ is the map $\mathcal{C}\colon\W\to\R$ defined as
		$$ \mathcal{C}(v_{1q},\dotsc,v_{kq},\vartheta_q^1,\dotsc,\vartheta_q^1, s^\alpha) = \vartheta_q^\alpha(v_{\alpha q})\,. $$
		\item The \textbf{canonical 1-forms} $\Theta^\alpha = \rho_2^\ast\:\theta^\alpha\in\Omega^1(\W)$.
		\item The \textbf{canonical 2-forms} $\Omega^\alpha = \rho_2^\ast\:\omega^\alpha = -\d\Theta^\alpha\in\Omega^2(\W)$.
		\item The \textbf{contact 1-forms} $\eta^\alpha = \d s^\alpha - \Theta^\alpha\in\Omega^1(\W)$. Notice that $\d\eta^\alpha = \Omega^\alpha$.
	\end{enumerate}
\end{dfn}

In natural coordinates of $\W$, these natural structures are written as
\begin{equation*}
	\Theta^\alpha = p_i^\alpha\d q^i\ ,\quad\Omega^\alpha = \d q^i\wedge\d p_i^\alpha\ ,\quad\eta^\alpha = \d s^\alpha - p_i^\alpha\d q^i\,.
\end{equation*}

The contact 1-forms $\eta^1,\dotsc,\eta^k$ define a $k$-precontact structure in the manifold $\W$. Notice that this is not a $k$-contact structure because conditions $(ii)$ and $(iii)$ on Definition \ref{dfn-kcontact-manifold} do not hold.

\begin{prop} There exists a family of Reeb vector fields $\Reeb_1,\dotsc,\Reeb_k\in\X(\W)$ such that
$$
	\begin{cases}
		i(\Reeb_\alpha)\d\eta^\beta = 0\,,\\
		i(\Reeb_\alpha)\eta^\beta = \delta_\alpha^\beta\,.
	\end{cases}
$$
\end{prop}

Notice that, since the manifold $\W$ is $k$-precontact, 
the family $(\Reeb_\alpha)$ of Reeb vector fields is not unique. 
In coordinates, $\Reeb_\alpha$ can be written as
\begin{equation}\label{eq:Reeb-local-expr}
	\Reeb_\alpha = \parder{}{s^\alpha} + (\Reeb_\alpha)^i_\beta\parder{}{v^i_\beta}\,,
\end{equation}
where $(\Reeb_\alpha)^i_\beta$ are arbitrary functions in $\W$.

\begin{dfn}
	Let $L\in\Cinfty(\oplus^k\T Q\times\R^k)$ be a Lagrangian function and let $\L=\rho_1^\ast L\colon\W\to\R$. We define the \textbf{Hamiltonian function} associated to $L$ by
	\begin{equation}\label{Hamiltonian-function}
		\H = \mathcal{C} - \L = p_i^\alpha v^i_\alpha - L(q^j,v^j_\alpha,s^\alpha)\in\Cinfty(\W)\,.
	\end{equation}
\end{dfn}

\begin{obs}{\rm
Notice that, 
since the manifold $\W$ along with the contact 1-forms $\eta^\alpha$ 
is a $k$-precontact manifold, 
$(\W,\eta^\alpha,\H)$ is a $k$-precontact Hamiltonian system.
}\end{obs}

\subsection{\texorpdfstring{$k$}--contact dynamical equations in the unified formulation}

\begin{dfn}
	The \textbf{Lagrangian-Hamiltonian problem} associated with the $k$-precontact system $(\W,\eta,\H)$ consists in finding the integral sections $\psi\colon\R^k\to\W$ of a $k$-vector field $\mathbf{Z} = (Z_1,\dotsc,Z_k)$ in $\W$ satisfying
	\begin{equation}\label{k-contact-fields-SR}
		\begin{cases}
			i(Z_\alpha)\d\eta^\alpha = \d\mathcal{H} - (\Lie_{\Reeb_\alpha}\mathcal{H})\eta^\alpha\,,\\
			i(Z_\alpha)\eta^\alpha = -\mathcal{H}\,,
		\end{cases}
	\end{equation}
	or, what is equivalent,
	$$ \begin{cases}
		\Lie_{Z_\alpha}\eta^\alpha = -(\Lie_{\Reeb_\alpha}\mathcal{H})\eta^\alpha\,,\\
		i(Z_\alpha)\eta^\alpha = -\mathcal{H}\,.
	\end{cases} $$
\end{dfn}

Given that $(\W,\eta^\alpha,\H)$ is a $k$-precontact Hamiltonian system, equations \eqref{k-contact-fields-SR} are not consistent everywhere in $\W$. Hence, we need to use the constraint algorithm described in Section \ref{singular} 
in order to find (if it exists) a final constraint submanifold of $\W$ 
where the existence of consistent solutions to equations \eqref{k-contact-fields-SR} is assured.

In natural coordinates $(q^i,v^i_\alpha,p_i^\alpha,s^\alpha)$ of $\W$, the local expression of a $k$-vector field $\mathbf{Z} = (Z_1,\dotsc,Z_k)$ in $\W$ is
$$ Z_\alpha = (Z_\alpha)^i\parder{}{q^i} + (Z_\alpha)^i_\beta\parder{}{v^i_\beta} + (Z_\alpha)^\beta_i\parder{}{p_i^\beta} + (Z_\alpha)^\beta\parder{}{s^\beta}\,. $$
Therefore, we have
\begin{align*}
	i(Z_\alpha)\d\eta^\alpha &= (Z_\alpha)^i\d p_i^\alpha - (Z_\alpha)_i^\alpha\d q^i\,,\\
	i(Z_\alpha)\eta^\alpha &= (Z_\alpha)^\alpha - p_i^\alpha (Z_\alpha)^i\,.
\end{align*}
Furthermore,
\begin{align*}
	\d\H &= v_\alpha^i\d p_i^\alpha + \left( p_i^\alpha - \parder{\L}{v_\alpha^i} \right)\d v_\alpha^i - \parder{\L}{q^i}\d q^i - \parder{\L}{s^\alpha}\d s^\alpha\,,\\
	(\Reeb_\alpha(\H))\eta^\alpha &= -\parder{\L}{s^\alpha}(\d s^\alpha - p_i^\alpha\d q^i)\,.
\end{align*}
Taking all this into account, the second equation \eqref{k-contact-fields-SR} gives
\begin{equation}\label{second}
	(Z_\alpha)^\alpha = \left((Z_\alpha)^i - v_\alpha^i\right)p_i^\alpha + L\circ\rho_1\,,
\end{equation}
and the first equation \eqref{k-contact-fields-SR} leads to the conditions
\begin{align}
	(Z_\alpha)^i &= v_\alpha^i & &(\mbox{coefficients in }\d p_i^\alpha)\,,\label{one}\\
	p_i^\alpha &= \parder{\L}{v_\alpha^i} = \parder{L}{v_\alpha^i}\circ\rho_1 & &(\mbox{coefficients in }\d v_\alpha^i)\,,\label{two}\\
	(Z_\alpha)_i^\alpha &= \parder{L}{q^i}\circ\rho_1 + p_i^\alpha\left(\parder{L}{s^\alpha}\circ\rho_1\right) & &(\mbox{coefficients in }\d q^i)\,.\label{three}
\end{align}
From these equations we have that:
\begin{itemize}
	\item Conditions \eqref{second} and \eqref{one} imply that $(Z_\alpha)^\alpha = L\circ\rho_1\,.$
	\item Equations \eqref{one} are the holonomy conditions. This means that the $k$-vector field $\mathbf{Z}$ is a \textsc{sopde}. Hence, as usual, we obtain straightforwardly the \textsc{sopde} condition from the Skinner--Rusk formalism. This is an important difference with the Lagrangian formalism, where we need to impose the second order condition in the case of singular Lagrangians.
	\item The algebraic equations \eqref{two} are consistency conditions which define a {\sl first constraint submanifold} $\W_1\hookrightarrow\W$. In fact, $\W_1$ is essentially the graph of $\F\L$:
	$$ \W_1 = \left\{ (v_q,\F\L(v_q))\in\W\ \vert\  v_q\in\oplus^k\T Q\times\R^k \right\}\,.$$
	This means that the Skinner--Rusk formalism includes the definition of the Legendre map as a consequence of the constraint algorithm.
\end{itemize}
Taking all this into account, a $k$-vector field $\mathbf{Z} = (Z_1,\dotsc,Z_k)$ is a solution to \eqref{k-contact-fields-SR} if $Z_\alpha$ has the form
$$ 
Z_\alpha = v^i_\alpha\parder{}{q^i} + (Z_\alpha)^i_\beta\parder{}{v^i_\beta} + (Z_\alpha)^\beta_i\parder{}{p_i^\beta} + (Z_\alpha)^\beta\parder{}{s^\beta}\qquad\mbox{(on $\W_1$)}\,, 
$$
with the restrictions
$$
	\begin{dcases}
		(Z_\alpha)^\alpha = \L\,,\\
		(Z_\alpha)_i^\alpha = \parder{\L}{q^i} + p_i^\alpha\parder{\L}{s^\alpha}\,.
	\end{dcases}
$$
Notice that the $k$-vector field $\mathbf{Z}$ 
does not depend on the arbitrary functions $(\Reeb_\alpha)^i_\beta$
chosen to define the Reeb vector fields 
according to \eqref{eq:Reeb-local-expr}.

At this point the constraint algorithm continues by demanding the tangency of $\mathbf{Z}$ 
to the first constraint submanifold $\W_1$. 
We denote by $\xi_j^\beta$ the constraint functions defining $\W_1$,
$$ \xi_j^\beta = p_j^\beta - \parder{\L}{v_\beta^j}\,. $$
Imposing the tangency conditions $X_\alpha(\xi_j^\beta)=0$ we obtain
\begin{equation}\label{tangency-condition}
	0 = X_\alpha(\xi_j^\beta) = X_\alpha\bigg(p_j^\beta - \parder{\L}{v_\beta^j}\bigg) = (Z_\alpha)_j^\beta - \parder{^2\L}{q^i\partial v_\beta^j}v^i_\alpha - \parder{^2\L}{v_\gamma^i\partial v_\beta^j}(Z_\alpha)_\gamma^i - \parder{^2\L}{s^\gamma\partial v_\beta^j}(Z_\alpha)^\gamma \qquad\mbox{(on $\W_1$)}\,,
\end{equation}
which partially determine the coefficients of the $k$-vector field~$\mathbf{Z}$.


It is interesting to point out that, in general, equations \eqref{k-contact-fields-SR} do not have a unique solution. Solutions to \eqref{k-contact-fields-SR} are given by
$$ 
(Z_1,\dotsc,Z_k) + (\ker\Omega^\sharp\cap\ker\eta^\sharp)\ , 
$$
where $\mathbf{Z}=(Z_1,\dotsc,Z_k)$ is a particular solution, $\Omega^\sharp$ is the morphism defined by
\begin{align*}
	\Omega^\sharp \colon \oplus^k\T\W &\to \T^\ast\W\\
	(Z_1,\dotsc,Z_k) &\mapsto \Omega^\sharp(Z_1,\dotsc,Z_k) = i(Z_\alpha)\d\eta^\alpha.
\end{align*}
and $\eta^\sharp$ is given by $\eta^\sharp(Z_1,\dotsc,Z_k) = \eta^\alpha(Z_\alpha)$.

Now we distinguish two cases:
\begin{itemize}
	\item If the Lagrangian function $\L$ is regular, equations \eqref{tangency-condition} allow us to compute the functions $(Z_\alpha)_\gamma^i$. Notice that, however, we do not have uniqueness of solutions to equations \eqref{k-contact-fields-SR}.
	\item If $\L$ is a singular Lagrangian, these equations establish some relations among the functions $(Z_\alpha)_\gamma^i$. In addition, some new constraints may appear defining a new constraint submanifold $\W_2\hookrightarrow\W_1\hookrightarrow\W$. We must now implement the constraint algorithm described in Section \ref{singular} in order to obtain a constraint submanifold (if it exists) where we can ensure the existence of solutions. 
\end{itemize}

\subsection{Recovering the Lagrangian and the Hamiltonian formalisms}

Consider the restriction of the projections $\rho_1\colon\W\to\oplus^k\T Q\times\R^k$, $\rho_2\colon\W\to\oplus^k\T^\ast Q\times\R^k$
restricted to $\W_1\subset\W$,
\begin{equation*}
	\rho_1^0\colon\W_1\to\oplus^k\T Q\times\R^k\ ,\qquad \rho_2^0\colon\W_1\to\oplus^k\T^\ast Q\times\R^k\,.
\end{equation*}
Since $\W_1$ is the graph of the Legendre transformation $\F\L$, it is clear that the projection $\rho_1^0$ is really a diffeomorphism.

Consider an integrable $k$-vector field $\mathbf{Z}=(Z_1,\dotsc,Z_k)$ solution to equations \eqref{k-contact-fields-SR}. Every integral section $\psi\colon\R^k\to\W$, given by
$\psi(t) = (\psi^i(t),\psi^i_\alpha(t), \psi_i^\alpha(t),\psi^\alpha(t))$,
is of the form
$$ \psi = (\psi_L,\psi_H)\,, $$
with $\psi_L = \rho_1\circ\psi\colon\R^k\to\oplus^k\T Q\times\R^k$, and if $\psi$ takes values in $\W_1$, we also have that $\psi_H = \F\L\circ\psi_L$:
$$
	\psi_H(t) = (\rho_2\circ\psi)(t)= (\psi^i(t),\psi_i^\alpha(t),\psi^\alpha(t))=\left( \psi^i(t), \parder{\L}{v_\alpha^i}(\psi_L(t)), \psi^\alpha(t) \right)=(\F\L\circ\psi_L)(t)\,,
$$
where we have used \eqref{two}. 
Notice that, in this way, 
we can always project from the Skinner--Rusk formalism onto the Lagrangian or the Hamiltonian formalisms 
by restricting to the first or second factor of the Pontryagin bundle~$\W$. 
In particular, relations \eqref{two} define the image of the Legendre transformation 
$\F\L(\oplus^k\T Q\times\R^k)\subset\oplus^k\T^\ast Q\times\R^k$. 
These relations are called \textsl{primary Hamiltonian constraints}.

The following theorem establishes how we can recover 
the Euler--Lagrange equations \eqref{k-contact-sections-lag} 
from the Skinner--Rusk formalism.

\begin{thm}
\label{T0}
	Consider an integrable $k$-vector field $\mathbf{Z}=(Z_1,\dotsc,Z_k)$ in $\W$, solution to equations \eqref{k-contact-fields-SR}. Let $\psi\colon\R^k\to\W_1\subset\W$ be an integral section of $\mathbf{Z}$ given by $\psi = (\psi_L,\psi_H)$, with $\psi_H = \F\L\circ\psi_L$. Then, $\psi_L$ is the first prolongation of the projected section $\phi = \tau_0\circ\rho_1^0\circ\psi\colon\R^k\to Q\times\R^k$, and $\phi$ is a solution to the Euler--Lagrange equations \eqref{k-contact-sections-lag}.
\end{thm}
\begin{proof}
	Consider an integral section $\psi(t) = \left( \psi^i(t), \psi_\alpha^i(t), \psi_i^\alpha(t), \psi^\alpha(t) \right)$ of the $k$-vector field $\mathbf{Z}$. Then, we have that
	\begin{equation}\label{sections-field}
		Z_\alpha(\psi(t)) = \parder{\psi^i}{t^\alpha}(t)\left.\parder{}{q^i}\right\vert_{\psi(t)} + \parder{\psi^i_\beta}{t^\alpha}(t)\left.\parder{}{v_\beta^i}\right\vert_{\psi(t)} + \parder{\psi_i^\beta}{t^\alpha}(t)\left.\parder{}{p_i^\beta}\right\vert_{\psi(t)} + \parder{\psi^\beta}{t^\alpha}(t)\left.\parder{}{s^\beta}\right\vert_{\psi(t)}\,.
	\end{equation}
	Now, from \eqref{second}, \eqref{one}, \eqref{two} and \eqref{sections-field} we get
	\begin{align}
		\parder{\psi^\alpha}{t^\alpha}(t) &= (L\circ\rho_1)(\psi(t)) = L(\psi_L(t))\,,\label{one-sections-L}\\
		\psi_i^\alpha(t) &= p_i^\alpha(\psi(t)) = \left( \parder{L}{v_\alpha^i}\circ\rho_1 \right)(\psi(t)) = \parder{L}{v_\alpha^i}(\psi_L(t))\,,\\
		\psi_\alpha^i(t) &= v_\alpha^i(\psi(t)) = (Z_\alpha^i)(\psi(t)) = \parder{\psi^i}{t^\alpha}(t)\,,\label{three-sections-L}\\
		\parder{\psi_i^\beta}{t^\alpha}(t) &= (Z_\alpha)^\beta_i(\psi(t))\,.\label{four-sections-L}
	\end{align}
	Using the conditions above and equation \eqref{three}, we obtain
	$$ \parder{\psi^\alpha_i}{t^\alpha}(t) = \left( \parder{L}{q^i}\circ\rho_1 \right)(\psi(t)) + p_i^\alpha(\psi(t))\left( \parder{L}{s^\alpha}\circ\rho_1 \right)(\psi(t))\,, $$
	and hence,
	$$ \parder{}{t^\alpha}\parder{L}{v^i_\alpha}(\psi_L(t)) = \parder{L}{q^i}(\psi_L(t)) + \parder{L}{v_i^\alpha}(\psi_L(t))\parder{L}{s^\alpha}(\psi_L(t))\,. $$
	$$ \psi_L = \left(\psi^i,\parder{\psi^i}{t^\alpha},\psi^\alpha\right)\,, $$
	It is clear that $\psi_L$ is the first prolongation of the map $\phi = \tau_0\circ\rho_1\circ\psi\colon\R^k\to Q\times\R^k$ given by $\phi = (\psi^i,\psi^\alpha)$, which is a solution to the Euler--Lagrange field equations \eqref{euler-lagrange-integral-maps}.
\end{proof}

Now we see how to recover the Hamilton field equations \eqref{Hamilton-sections} from the Skinner--Rusk formalism.

\begin{thm}
\label{T1}
	Let $\mathbf{Z}=(Z_1,\dotsc,Z_k)$ be an integrable $k$-vector field in $\W$ solution to equations \eqref{k-contact-fields-SR} and $\psi\colon\R^k\to\W_1\subset\W$ be an integral section of $\mathbf{Z}$ given by $\psi = (\psi_L,\psi_H)$, with $\psi_H = \F\L\circ\psi_L$. If the Lagrangian $L$ is regular, $\psi_H$ is a solution to the Hamilton field equations \eqref{Hamilton-sections}, where the Hamiltonian function $H$ is given by $E_L = H\circ\F L$.
\end{thm}
\begin{proof}
	We have that $L$ is a regular Lagrangian and hence, $\F\L$ is a local diffeomorphism. Then, for every point $p\in\oplus^k\T Q\times\R^k$, there exists an open subset $U\subset\oplus^k\T Q\times\R^k$ containing the point $p$ such that the restriction $\restr{\F\L}{U}\colon U\to\F\L(U)$ is a diffeomorphism. Using this, we can define a function $\tilde H = \restr{E_\L}{U}\circ(\restr{\F\L}{U})^{-1}$. From now on, we will consider that the maps $E_\L$ and $\F\L$ are restricted to the open set $U$.
	Now, using that $E_\L = \tilde H\circ\F\L$, it is clear that
	\begin{equation}\label{Hamilton-sections-U}
		\begin{dcases}
			\parder{\tilde H}{p^\alpha_i}\circ\F\L = v_\alpha^i\,,\\
			\parder{\tilde H}{q^i} \circ\F\L = -\parder{L}{q^i}\,.
		\end{dcases}
	\end{equation}
	We consider now the subset $V = \psi_L^{-1}(U)\subset\R^k$ and restrict $\psi$ to $V$, so we have
	\begin{equation*}
		\begin{array}{rcl}
			\restr{\psi}{V}\colon V\subset\R^k & \to & U \oplus_{\R^k}\F\L(U)\\
			t & \mapsto & (\psi_L(t),\psi_H(t)) = (\psi_L(t),(\F\L\circ\psi_L)(t))
		\end{array}
	\end{equation*}
	Taking into account \eqref{three}, \eqref{three-sections-L}, \eqref{four-sections-L} and \eqref{Hamilton-sections-U},
	\begin{align*}
		\parder{\tilde H}{p_i^\alpha}(\psi_H(t)) &= \left( \parder{\tilde H}{p_i^\alpha}\circ\F\L \right)(\psi_L(t)) = v_\alpha^i(\psi_L(t)) = \parder{\psi^i}{t^\alpha}(t)\,,\\
		\parder{\tilde H}{q^i}(\psi_H(t)) &= \left( \parder{\tilde H}{q^i}\circ\F\L \right)(\psi_L(t)) = -\parder{L}{q^i}(\psi_L(t)) = -\left( \parder{L}{q^i}\circ\rho_1 \right)(\psi(t))\\
		&= \left( p_i^\alpha\left( \parder{L}{s^\alpha}\circ\rho_1 \right) - (Z_\alpha)^\alpha_i \right)(\psi(t)) = p_i^\alpha\left( \parder{L}{s^\alpha}\circ\rho_1 \right)(\psi(t)) - (Z_\alpha)^\alpha_i(\psi(t))\\
		&= p_i^\alpha \parder{L}{s^\alpha}(\psi_L(t)) - \parder{\psi^\alpha_i}{t^\alpha}(t) = -p_i^\alpha\parder{E_\L}{s^\alpha}(\psi_L(t)) - \parder{\psi^\alpha_i}{t^\alpha}(t)\\
		&= -p_i^\alpha\parder{(\tilde H\circ\F\L)}{s^\alpha}(\psi_L(t)) - \parder{\psi^\alpha_i}{t^\alpha}(t) = -p_i^\alpha\parder{\tilde H}{s^\alpha}(\psi_H(t)) - \parder{\psi^\alpha_i}{t^\alpha}(t)\,,
	\end{align*}
	and then
	$$
		\begin{dcases}
			\parder{\psi^i}{t^\alpha}(t) = \parder{\tilde H}{p_i^\alpha}(\psi_H(t))\,,\\
			\parder{\psi^\alpha_i}{t^\alpha}(t) = -\left( \parder{\tilde H}{q^i} + p_i^\alpha\parder{\tilde H}{s^\alpha} \right)(\psi_H(t))\,.
		\end{dcases}
	$$
Finally, considering equation \eqref{one-sections-L}, we deduce
	$$ \parder{\psi^\alpha}{t^\alpha}(t) = L\circ\psi_L(t) = p_i^\alpha\parder{\psi^i}{t^\alpha}(t) - \tilde H\circ\psi_H = \left( p_i^\alpha\parder{\tilde H}{p_i^\alpha} - \tilde H \right)(\psi_H(t))\,. $$
Hence, we have that $\psi_H$ is a solution of the Hamilton field equations \eqref{Hamilton-sections} on $V$.
\end{proof}

We have seen that we can recover the Euler--Lagrange field equations and Hamilton field equations from the Skinner--Rusk formalism. 
Conversely, we have the following result:

\begin{thm}
\label{T2}
	Let $L\in\Cinfty(\oplus^k\T Q\times\R^k)$ be a regular Lagrangian function and consider a $k$-vector field $\mathbf{X}=(X_1,\dotsc,X_k)$
	in $\oplus^k\T Q\times\R^k$, solution to the $k$-contact Lagrangian equations \eqref{k-contact-fields-lag}. Then, the $k$-vector field $\mathbf{Z}=(Z_\alpha)$ in $\W$ defined as $Z_\alpha = (\mathrm{Id}_{\oplus^k\T Q\times\R^k}\times\F\L)_\ast(X_\alpha)$ is a solution to equations \eqref{k-contact-fields-SR}. Moreover, if $\psi_L\colon\R^k\to\oplus^k\T Q\times\R^k$ is an integral section of $\mathbf{X}$, $\psi=(\psi_L,\F\L\circ\psi_L)\colon\R^k\to\W$ is an integral section of $\mathbf{Z}$.
\end{thm}
\begin{proof}
	Consider a regular Lagrangian function $L\in\Cinfty(\oplus^k\T Q\times\R^k)$ and let $\mathbf{X}=(X_1,\dotsc,X_k)$ be a $k$-vector field in $\oplus^k\T Q\times\R^k$ solution to equations \eqref{k-contact-fields-lag}. Hence, $X_\alpha$ is written in coordinates as
	$$ X_\alpha = v_\alpha^i\parder{}{q^i} + (X_\alpha)^i_\beta\parder{}{v_\beta^i} + (X_\alpha)^\beta\parder{}{s^\beta}\,, $$
	where the functions $(X_\alpha)^\beta$ and $(X_\alpha)^i_\beta$ satisfy the conditions
	\begin{align}
		(X_\alpha)^\alpha &= L\,,\label{cond-1}\\
		-\parder{L}{q^i} + \parder{^2L}{s^\beta\partial v_\alpha^i}(X_\alpha)^\beta + \parder{^2L}{q^j\partial v_\alpha^i}v_\alpha^j + \parder{^2L}{v^j_\beta\partial v^i_\alpha}(X_\alpha)^j_\beta &= \parder{L}{s^\alpha}\parder{L}{v_\alpha^i}\label{cond-2}\,.
	\end{align}
	Now, using the coordinate expression \eqref{legendre-coords} of the Legendre map $\F L$ and taking into account that $Z_\alpha = (\mathrm{Id}_{\oplus^k\T Q\times\R^k}\times\F\L)_\ast(X_\alpha)$, we have
	\begin{equation}\label{local-expr}
		Z_\alpha = v_\alpha^i\parder{}{q^i} + (X_\alpha)^i_\beta\parder{}{v_\beta^i} + \left(v^j_\alpha\parder{^2L}{q^j\partial v_\gamma^i} + (X_\alpha)^j_\beta\parder{^2L}{v^j_\beta\partial v_\gamma^i} + (X_\alpha)^\beta\parder{^2L}{s^\beta\partial v^i_\gamma}\right)\parder{}{p_i^\gamma} + (X_\alpha)^\beta\parder{}{s^\beta}\,.
	\end{equation}
	From \eqref{cond-1}, \eqref{cond-2} and \eqref{local-expr}, it is clear that $\mathbf{Z} = (Z_\alpha)$ fulfills conditions \eqref{second}, \eqref{one}, \eqref{three} and \eqref{tangency-condition} and hence, the $k$-vector field $\mathbf{Z}$ is a solution of \eqref{k-contact-fields-SR} tangent to $\W_1$.

	It is also clear from the definition of integral section that $\psi = (\psi_L,\F\L\circ\psi_L)$ is an integral section of $\mathbf{Z}$.
\end{proof}

\begin{obs}
{\rm In the case of singular Lagrangians,
the results in Theorems \ref{T0}, \ref{T1}, and \ref{T2}
hold on the corresponding final constraint submanifolds
of the Skinner--Rusk, Lagrangian and Hamiltonian formalisms.}
$$
\xymatrix{
\ & \ & {\cal W} \ar@/_1.3pc/[ddll]_{\rho_1} \ar@/^1.3pc/[ddrr]^{\rho_2} & \ & \ \\
\ & \ & {\cal W}_1 \ar[dll]_{\rho_1^0} \ar[drr]^{\rho_2^0} \ar@{^{(}->}[u] & \ & \ \\
\oplus^k\T Q\times\R^k \ar[rrrr]^<(0.30){{\cal F}L}|(.49){\hole}
\ar[drrrr]|(.49){\hole} 
& \ & \ & \ & \oplus^k\T^\ast Q\times\R^k  \\
\ & \ &  {\cal W}_f \ar@{^{(}->}[uu] \ar[dll]\ar[drr]  & \ & {\cal P} \ar@{^{(}->}[u] \\
S_f \ar@{^{(}->}[uu] \ar[rrrr] & \ & \ & \ & P_f \ar@{^{(}->}[u] \\
}
$$
\end{obs}


\section{Examples}
\label{4}

\subsection{1-dimensional wave equation with dissipation}
\label{exmpl-wave}

In this example we study a vibrating string with friction. We begin by considering the Lagrangian function $L\colon\oplus^2\T\R\to\R$ defined by
$$ L(u,u_x,u_t) = \frac{1}{2}\rho u_t^2 - \frac{1}{2}\tau u_x^2\,, $$
which originates the one-dimensional wave equation
$$ 
\parder{^2u}{t^2} = c^2\parder{^2u}{x^2}\quad, \quad
\mbox{\rm (where $c^2=\tau/\rho$)} \ .
$$
Adding to this wave equation a dissipation term proportional to the speed of an element of the string we get a simple model of a vibrating string with dissipation of energy, that is, with friction:
\begin{equation}\label{damped-wave}
	\parder{^2u}{t^2} - c^2\parder{^2u}{x^2} + \gamma\parder{u}{t} = 0 \,,
\end{equation}
where $\gamma>0$ is the damping constant. This equation can be obtained from the Lagrangian
$$ \L = L - \gamma s^t $$
defined in the 2-contact manifold $\oplus^2\T\R\times\R^2$ endowed with coordinates $(u, u_x, u_t, s^x, s^t)$.
The Hamiltonian and the Lagrangian formalisms for this model were analyzed in \cite{GGMRR-2019,GGMRR-2020}, respectively.

Next we apply the Skinner--Rusk formalism to this system and we see how to recover the damped wave equation \eqref{damped-wave} and both the Lagrangian and the Hamiltonian formalisms.
Consider the extended Pontryagin bundle
$$ \W = \oplus^2\T\R\times_\R\oplus^2\T^\ast\R\times\R^2 $$
endowed with canonical coordinates $(u, u_x, u_t, p^x, p^t, s^x, s^t)$. In this bundle, the coupling function is
$$ \mathcal{C} = p^xu_x + p^tu_t\,, $$
and we have the canonical forms
\begin{align*}
	\Theta^1 = p^x\d u\quad &,\quad
		\Omega^1 = -\d\Theta^1 = \d u\wedge\d p^x\,,\\
	\Theta^2 = p^t\d u\quad &,\quad
	\Omega^2 = -\d\Theta^2 = \d u\wedge\d p^t\,,
\end{align*}
and the canonical contact forms
$$
	\eta^1 = \d s^x - p^x\d u\quad,\quad
	\eta^2 = \d s^t - p^t\d u\ .
$$
We can take the vector fields
$$ \Reeb_1 = \parder{}{s^x}\ ,\quad \Reeb_2 = \parder{}{s^t} $$
as Reeb vector fields.
Given the Lagrangian function $\L\colon\oplus^2\T\R\times\R^2\to\R$ defined by
$$ \L(u, u_x, u_t, s^x, s^t) = \frac{1}{2}\rho u_t^2 - \frac{1}{2}\tau u_x^2 - \gamma s^t\,, $$
we can construct the Hamiltonian function $\H = \mathcal{C} - \L$, which in coordinates reads
$$ \H = p^xu_x + p^tu_t - \frac{1}{2}\rho u_t^2 + \frac{1}{2}\tau u_x^2 + \gamma s^t\,. $$
To solve the Lagrangian--Hamiltonian problem for the 2-precontact Hamiltonian system $(\W,\eta^\alpha,\H)$ means to find a 2-vector field $\mathbf{Z} = (Z_1,Z_2)$ in $\W$ satisfying equations \eqref{k-contact-fields-SR}.
For our Hamiltonian function $\H$, we have
$$ \d\H - \Reeb_\alpha(\H)\eta^\alpha = u_x\d p^x + u_t\d p^t + (p^x + \tau u_x)\d u_x + (p^t - \rho u_t)\d u_t + \gamma p^t\d u\,. $$
Let $\mathbf{Z} = (Z_\alpha)$ be a 2-vector field with local expression
$$ Z_\alpha = f_\alpha\parder{}{u} + F_{\alpha 1}\parder{}{u_x} + F_{\alpha 2}\parder{}{u_t} + G_\alpha^1\parder{}{p^x} + G_\alpha^2\parder{}{p^t} + g_\alpha^1\parder{}{s^x} + g_\alpha^2\parder{}{s^t}\,. $$
Now,
$$ i(Z_\alpha)\d\eta^\alpha = f_1\d p^x + f_2\d p^t - (G_1^1 + G^2_2)\d u\,, $$
and hence, the first equation in \eqref{k-contact-fields-SR} gives the conditions
\begin{align*}
	G_1^1 + G_2^2 &= -\gamma p^t & & \mbox{(coefficients in $\d u$)}\,,\\
	p^x &= -\tau u_x & & \mbox{(coefficients in $\d u_x$)}\,,\\
	p^t &= \rho u_t & & \mbox{(coefficients in $\d u_t$)}\,,\\
	f_1 &= u_x & & \mbox{(coefficients in $\d p^x$)}\,,\\
	f_2 &= u_t & & \mbox{(coefficients in $\d p^t$)}\,.
\end{align*}
Notice that combining the first three conditions we recover the damped wave equation \eqref{damped-wave}. Furthermore, the last two equations are the holonomy conditions.
The second equation in \eqref{k-contact-fields-SR} gives the condition
$$ g_1^1 + g_2^2 = \frac{1}{2}\rho u_t^2 - \frac{1}{2}\tau u_x^2 - \gamma s^t = \L\,. $$
In addition, we have obtained the constraints
$$ \xi_1 = p^x + \tau u_x = 0\quad,\quad\xi_2 = p^t - \rho u_t = 0 \ , $$
which define the submanifold $\W_1\hookrightarrow\W$. Imposing the tangency of the 2-vector field $\mathbf{Z}$ to the submanifold $\W_1$ we get the conditions
\begin{align*}
	0 = Z_1(\xi_1) = G_1^1 + \tau F_{11}\quad &,\quad
	0 = Z_2(\xi_1) = G_2^1 + \tau F_{21}\,,\\
	0 = Z_1(\xi_2) = G_1^2 - \rho F_{12}\quad &,\quad
	0 = Z_2(\xi_2) = G_2^2 - \rho F_{22}\,,
\end{align*}
which determine partially some of the arbitrary functions and no new constraints appear, so the constraint algorithm finishes with the submanifold $\W_f = \W_1$, giving the solutions $\mathbf{Z}=(Z_1,Z_2)$ with
\begin{align*}
	Z_1 &= u_x\parder{}{u} - \frac{G_1^1}{\tau}\parder{}{u_x} + \frac{G_1^2}{\rho}\parder{}{u_t} + G_1^1\parder{}{p^x} + G_1^2\parder{}{p^t} + g_1^1\parder{}{s^x} + g_1^2\parder{}{s^t}\,,\\
	Z_2 &= u_t\parder{}{u} - \frac{G_2^1}{\tau}\parder{}{u_x} - \frac{G_1^1 + \gamma p^t}{\rho}\parder{}{u_t} + G_2^1\parder{}{p^x} - (G_1^1 + \gamma p^t)\parder{}{p^t} + g_2^1\parder{}{s^x} + (\L - g_1^1)\parder{}{s^t}\,,
\end{align*}
where $G_1^1,G_1^2,G_2^1,g_1^1,g_1^2,g_2^1$ are arbitrary functions.

It is important to point out that we can project on each factor of the product manifold $\W = \oplus^2\T\R\times_\R\oplus^2\T^\ast\R\times\R^2$ with the projections $\rho_1$ and $\rho_2$ to recover the Lagrangian and Hamiltonian formalisms. In the Lagrangian formalism we have the holonomic 2-vector field $\mathbf{X} = (X_1,X_2)$ given by
\begin{align*}
	X_1 &= u_x\parder{}{u} + F_{11}\parder{}{u_x} + F_{12}\parder{}{u_t} + g_1^1\parder{}{s^x} + g_1^2\parder{}{s^t}\,,\\
	X_2 &= u_t\parder{}{u} + F_{21}\parder{}{u_x} + \left(\frac{\tau}{\rho} F_{11} - \gamma u_t \right)\parder{}{u_t} + g_2^1\parder{}{s^x} + (\L-g_1^1)\parder{}{s^t}\,,
\end{align*}
where $F_{11},F_{12},F_{21},g_1^1,g_1^2,g_2^1$ are arbitrary functions. On the other side, in the Hamiltonian formalism we have the Hamiltonian 2-vector field $\mathbf{Y} = (Y_1,Y_2)$ given by
\begin{align*}
	Y_1 &= \frac{-p^x}{\tau}\parder{}{u} + G_1^1\parder{}{p^x} + G_1^2\parder{}{p^t} + g_1^1\parder{}{s^x} + g_1^2\parder{}{s^t}\,,\\
	Y_2 &= \frac{p^t}{\rho}\parder{}{u} + G_2^1\parder{}{p^x} - (G_1^1 + \gamma p^t)\parder{}{p^t} + g_2^1\parder{}{s^x} + \left(\frac{(p^t)^2}{2\rho} - \frac{(p^x)^2}{2\tau} - \gamma s^t - g_1^1\right)\parder{}{s^t}\,,
\end{align*}
where $G_1^1,G_1^2,G_2^1,g_1^1,g_1^2,g_2^1$ are arbitrary functions.

\subsection{From the massive scalar field to the telegrapher's equation}

The voltage and current on a uniform electrical transmission line can be described by the {\sl telegrapher's equations}
\cite[p.\,306]{HaBu}
\cite[p.\,653]{Salsa}:
\begin{equation*}
    \begin{dcases}
        \parder{V}{x} = -L\parder{I}{t} - RI\,,\\
        \parder{I}{x} = -C\parder{V}{t} - GV\,.
    \end{dcases}
\end{equation*}
From these equations, one can easily deduce the uncoupled system
\begin{equation*}
    \begin{dcases}
        \parder{^2V}{x^2} = LC\parder{^2V}{t^2} + (LG + RC)\parder{V}{t} + RGV\,,\\
        \parder{^2I}{x^2} = LC\parder{^2I}{t^2} + (LG + RC)\parder{I}{t} + RGI\,.\\
    \end{dcases}
\end{equation*}
These two identical equations are also known as telegrapher's equations.
Both of them can be written as
\begin{equation}\label{telegraph-equation}
	\square u + \gamma\parder{u}{t} + m^2 u = 0\,,
\end{equation}
where $\square$ is the d'Alembertian operator in 1+1 dimensions,
and $\gamma$ and $m^2$ are appropriate constants.
In this way, 
telegrapher's equation can be seen as a kind of modified Klein--Gordon equation.
Indeed, we will show that this equation can be obtained by adding a dissipative term to the Klein--Gordon Lagrangian, 
and treating it as a 4-contact Lagrangian. 

\subsubsection*{The Klein--Gordon equation}

One of the most important equations in field theory, 
either classical or quantum,
is the so-called Klein--Gordon equation
\cite[p.\,108]{IZ},
which can be written
\begin{equation}\label{klein-gordon}
(\square + m^2) \phi = 0
\,.
\end{equation}
Here $\phi$ is a scalar field in Minkowski space 
and $m^2$ a constant parameter.
This equation derives from the Lagrangian
\begin{equation}
L = \frac12 (\partial\phi)^2 - \frac12 m^2 \phi^2
\,.
\end{equation}
This can be slightly generalized to include a potential,
$
L = \frac12 (\partial\phi)^2 - V(\phi)
$,
but we will stick ourselves to the simplest case.

Since this Lagrangian is autonomous
and the space-time is Minkowski space $\R^4$,
it can be described as a 4-symplectic field theory.
We will use space-time coordinates
$(x^0,x^1,x^2,x^3)$
and write $q$ the field variable
and $v_i = \partial q/\partial x^i$
its velocities.
With these notations the Lagrangian 
$L \colon \oplus^4 \T\R \to \R$
is
\begin{equation}
\label{klein-gordon-lagrangian}
L(q,v_0,v_1,v_2,v_3) = 
\frac{1}{2}\left(v_0^2 - v_1^2 - v_2^2 - v_3^2\right)
-
\frac{1}{2} m^2 q^2 
\end{equation}
and the Klein--Gordon equation reads as
$$ 
\parder{^2\phi}{(x^0)^2} - 
\parder{^2\phi}{(x^1)^2} -
\parder{^2\phi}{(x^2)^2} - 
\parder{^2\phi}{(x^4)^2} +
m^2\phi =
0
\,. 
$$

\subsubsection*{From the Klein--Gordon to telegrapher's equation}

Consider now the contactified Lagrangian 
$\L \colon \oplus^4 \T\R \times \R^4 \to \R$ 
given by
\begin{equation}
\label{klein-gordon-lagrangian-dis}
\L(q, v_\alpha, s^\alpha) = 
L(q, v_\alpha) + \gamma_\mu s^\mu = 
\frac{1}{2} \left(v_0^2 - v_1^2 - v_2^2 - v_3^2\right)
- \frac{1}{2} m^2 q^2 + \gamma_\mu s^\mu
\,,
\end{equation}
defined in the 4-contact manifold 
$\oplus^4\T\R\times\R^4$, 
where $L$ is the Klein--Gordon Lagrangian \eqref{klein-gordon-lagrangian} 
and 
$\gamma = (\gamma_\mu)\in\R^4$ is a constant vector.

Consider the extended unified bundle
$\W = \oplus^4\T\R\times_\R\oplus^4\T^\ast\R\times\R^4$
equipped with canonical coordinates 
$(q, v_1,v_2,v_3,v_4,p^1,p^2,p^3,p^4,s^1,s^2,s^3,s^4)$. 
In this extended bundle we have the coupling function
$$ 
\C = p^0 v_0 + p^1 v_1 + p^2v_2 + p^3v_3
\,, 
$$
the canonical forms
\begin{align*}
\Theta^0 &= p^0\d q\ , \qquad 
\eta^0 = \d s^0 - p^0\d q\ ,\qquad 
\Omega^0 = -\d\Theta^0 = \d q\wedge\d p^0 = \d\eta^0\,,
\\
\Theta^1 &= p^1\d q\ , \qquad 
\eta^1 = \d s^1 - p^1\d q\ ,\qquad 
\Omega^1 = -\d\Theta^1 = \d q\wedge\d p^1 = \d\eta^1\,, 
\\
\Theta^2 &= p^2\d q\ , \qquad 
\eta^2 = \d s^2 - p^2\d q\ ,\qquad 
\Omega^2 = -\d\Theta^2 = \d q\wedge\d p^2 = \d\eta^2\,, 
\\
\Theta^3 &= p^3\d q\ , \qquad 
\eta^3 = \d s^3 - p^3\d q\ ,\qquad 
\Omega^3 = -\d\Theta^3 = \d q\wedge\d p^3 = \d\eta^3\,. 
\end{align*}
With the Lagrangian function \eqref{klein-gordon-lagrangian-dis}, we can construct the Hamiltonian function $\H = \C - \L$
which, in coordinates, reads
$$ 
\H = 
p^\alpha v_\alpha -
\frac{1}{2} \left(v_0^2 - v_1^2 - v_2^2 - v_3^2\right)
+ \frac{1}{2} m^2 q^2 
- \gamma_\mu s^\mu
\,.
$$
Now to solve the Lagrangian--Hamiltonian problem for the 4-precontact Hamiltonian system 
$(\W,\eta^\alpha,\H)$ 
we have to find a 4-vector field 
$\mathbf{Z} = (Z_0,Z_1,Z_2,Z_3)$ in $\W$ 
which satisfies equations \eqref{k-contact-fields-SR}. 
The vector fields
$$ 
\Reeb_\alpha = \parder{}{s^\alpha}\,, 
$$
for $\alpha = 0,1,2,3$
are Reeb vector fields of $\W$. 
For our Hamiltonian function $\H$, we have
\begin{equation}
\label{right-hand-side}
\d\H - \Reeb_\alpha(\H)\eta^\alpha = 
(p^0 - v_0)\d v_0 + 
(p^1 + v_1)\d v_1 + 
(p^2 + v_2)\d v_2 + 
(p^3 + v_3)\d v_3 + 
v_\alpha\d p^\alpha + 
(m^2q - \gamma_\mu p^\mu)\d q
\,.
\end{equation}
Consider a 4-vector field 
$\mathbf{Z} = (Z_\alpha)$ 
in $\W$ with local expression
$$ 
Z_\alpha = 
f_\alpha\parder{}{q} + 
F_{\alpha\beta}\parder{}{v_\beta} + 
G_\alpha^\beta\parder{}{p^\beta} + 
g_\alpha^\beta\parder{}{s^\beta}\,. 
$$
Computing the left-hand side of the first equation in \eqref{k-contact-fields-SR}, 
we have
$$ 
i(Z_\alpha)\d\eta^\alpha = 
f_\alpha\d p^\alpha - G_\alpha^\alpha\d q\,, 
$$
and, equating with \eqref{right-hand-side}, we obtain the conditions
\begin{align}
G_\alpha^\alpha &= -m^2q + \gamma_\mu p^\mu & & \mbox{(coefficients in $\d q$)}\,,\label{1-kg}
\\
p^0 &= v_0 & & 
\mbox{(coefficients in $\d v_0$)}\,,\label{2-kg}
\\
p^1 &= -v_1 & & 
\mbox{(coefficients in $\d v_1$)}\,,\label{3-kg}
\\
p^2 &= -v_2 & & 
\mbox{(coefficients in $\d v_2$)}\,,\label{4-kg}
\\
p^3 &= -v_3 & & 
\mbox{(coefficients in $\d v_3$)}\,,\label{5-kg}
\\
f_\alpha &= v_\alpha & & 
\mbox{(coefficients in $\d p^\alpha$)}\,.\label{6-kg}
\end{align}
Notice that the last equation is the holonomy condition, which is recovered from the unified formalism.
Furthermore, the second equation in \eqref{k-contact-fields-SR} gives the condition
$$ g_\alpha^\alpha = \L\,. $$
Moreover, we have obtained the constraints
$$
	\xi_0 = p^0 - v_0 = 0\quad ,\quad
	\xi_1 = p^1 + v_1 = 0\quad ,\quad
	\xi_2 = p^2 + v_2 = 0\quad ,\quad
	\xi_3 = p^3 + v_3 = 0\,,
$$
defining the submanifold $\W_1\hookrightarrow\W$. If we impose the tangency of the 4-vector field $\mathbf{Z}$ to this submanifold, we obtain the conditions
\begin{align*}
	0 = Z_\alpha(\xi_0) = G_\alpha^0 - F_{\alpha 0}\quad &,\quad
	0 = Z_\alpha(\xi_1) = G_\alpha^1 + F_{\alpha 1}\,,\\
	0 = Z_\alpha(\xi_2) = G_\alpha^2 + F_{\alpha 2}\quad &,\quad
	0 = Z_\alpha(\xi_3) = G_\alpha^3 + F_{\alpha 3}\,.
\end{align*}
These conditions partially determine some of the arbitrary functions and no new constraints appear. 
Hence, the constraint algorithm finishes with the submanifold $\W_f = \W_1$ 
and gives the solutions 
$\mathbf{Z} = (Z_0,Z_1,Z_2,Z_3)$, where
\begin{align*}
	Z_0 =& \ v_0\parder{}{q} + \left(-m^2 q + \gamma_\mu p^\mu - G_1^1 - G_2^2 - G_3^3\right)\parder{}{v_0} - G_0^1\parder{}{v_1} - G_0^2\parder{}{v_2} + G_0^3\parder{}{v_3}\\
	& + \left(-m^2 q + \gamma_\mu p^\mu - G_1^1 - G_2^2 - G_3^3\right)\parder{}{p^0} + G_0^1\parder{}{p^1} + G_0^2\parder{}{p^2} + G_0^3\parder{}{p^3} \\
	& + \left(\L - g_1^1 - g_2^2 - g_3^3\right)\parder{}{s^0} + g_0^1\parder{}{s^1} + g_0^2\parder{}{s^2} + g_0^3\parder{}{s^3}\,,\\
	Z_1 =& \ v_1\parder{}{q} + G_1^0\parder{}{v_0} - G_1^1\parder{}{v_1} - G_1^2\parder{}{v_2} - G_1^3\parder{}{v_3} + G_1^0\parder{}{p^0} + G_1^1\parder{}{p^1} + G_1^2\parder{}{p^2} + G_1^3\parder{}{p^3} \\
	& + g_1^0\parder{}{s^0} + g_1^1\parder{}{s^1} + g_1^2\parder{}{s^2} + g_1^3\parder{}{s^3}\,,
\end{align*}	
\begin{align*}
	Z_2 =& \ v_2\parder{}{q} + G_2^0\parder{}{v_0} - G_2^1\parder{}{v_1} - G_2^2\parder{}{v_2} - G_2^3\parder{}{v_3} + G_2^0\parder{}{p^0} + G_2^1\parder{}{p^1} + G_2^2\parder{}{p^2} + G_2^3\parder{}{p^3} \\
	& + g_2^0\parder{}{s^0} + g_2^1\parder{}{s^1} + g_2^2\parder{}{s^2} + g_2^3\parder{}{s^3}\,,\\
	Z_3 =& \ v_3\parder{}{q} + G_3^0\parder{}{v_0} - G_3^1\parder{}{v_1} - G_3^2\parder{}{v_2} - G_3^3\parder{}{v_3} + G_3^0\parder{}{p^0} + G_3^1\parder{}{p^1} + G_3^2\parder{}{p^2} + G_3^3\parder{}{p^3} \\
	& + g_3^0\parder{}{s^0} + g_3^1\parder{}{s^1} + g_3^2\parder{}{s^2} + g_3^3\parder{}{s^3}\,,
\end{align*}
where $G_\alpha^\beta,g_\alpha^\beta$, for $(\alpha,\beta)\in\left(\{0,1,2,3\}\times\{0,1,2,3\}\right)\setminus\{(0,0)\}$, are arbitrary functions.


Now we can project onto each factor of the manifold $\W$ using the projections $\rho_1,\rho_2$ to recover the Lagrangian and Hamiltonian formalisms. In the Lagrangian formalism we obtain the holonomic 4-vector field $\mathbf{X} = (X_0,X_1,X_2,X_3)$ given by
\begin{align*}
	X_0 =& \ v_0\parder{}{q} + \left( -m^2 q + \gamma_0v_0 - \gamma_1v_1 - \gamma_2v_2 - \gamma_3v_3 + F_1^1 + F_2^2 + F_3^3 \right)\parder{}{v_0}  \\
	&+ F_{01}\parder{}{v_1} + F_{02}\parder{}{v_2} + F_{03}\parder{}{v_3}
	 + \left( \L - g_1^1 - g_2^2 - g_3^3 \right)\parder{}{s^0} + g_0^1\parder{}{s^1} + g_0^2\parder{}{s^2} + g_0^3\parder{}{s^3}\,,\\
	X_1 =& \ v_1\parder{}{q} + F_{10}\parder{}{v_0} + F_{11}\parder{}{v_1} + F_{12}\parder{}{v_2} + F_{13}\parder{}{v_3} + g_1^0\parder{}{s^0} + g_1^1\parder{}{s^1} + g_1^2\parder{}{s^2} + g_1^3\parder{}{s^3}\,,\\
	X_2 =& \ v_2\parder{}{q} + F_{20}\parder{}{v_0} + F_{21}\parder{}{v_1} + F_{22}\parder{}{v_2} + F_{23}\parder{}{v_3} + g_2^0\parder{}{s^0} + g_2^1\parder{}{s^1} + g_2^2\parder{}{s^2} + g_2^3\parder{}{s^3}\,,\\
	X_3 =& \ v_3\parder{}{q} + F_{30}\parder{}{v_0} + F_{31}\parder{}{v_1} + F_{32}\parder{}{v_2} + F_{33}\parder{}{v_3} + g_3^0\parder{}{s^0} + g_3^1\parder{}{s^1} + g_3^2\parder{}{s^2} + g_3^3\parder{}{s^3}\,,
\end{align*}
where $F_\alpha^\beta,g_\alpha^\beta$ for $(\alpha,\beta)\in\left(\{0,1,2,3\}\times\{0,1,2,3\}\right)\setminus\{(0,0)\}$, are arbitrary functions. In the Hamiltonian counterpart, we get the Hamiltonian 4-vector field $\mathbf{Y} = (Y_0,Y_1,Y_2,Y_3)$ given by
\begin{align*}
	Y_0 =& \ v_0\parder{}{q} + \left(-m^2q + \gamma_\mu p^\mu - G_1^1 - G_2^2 - G_3^3\right)\parder{}{p^0} + G_0^1\parder{}{p^1} + G_0^2\parder{}{p^2} + G_0^3\parder{}{p^3} \\
	& + \left(\L - g_2^2 - g_3^3 - g_4^4\right)\parder{}{s^0} + g_0^1\parder{}{s^1} + g_0^2\parder{}{s^2} + g_0^3\parder{}{s^3}\,,\\
	Y_1 =& \ v_1\parder{}{q} + G_1^0\parder{}{p^0} + G_1^1\parder{}{p^1} + G_1^2\parder{}{p^2} + G_1^3\parder{}{p^3} + g_1^0\parder{}{s^0} + g_1^1\parder{}{s^1} + g_1^2\parder{}{s^2} + g_1^3\parder{}{s^3}\,,\\
	Y_2 =& \ v_2\parder{}{q} + G_2^0\parder{}{p^0} + G_2^1\parder{}{p^1} + G_2^2\parder{}{p^2} + G_2^3\parder{}{p^3} + g_2^0\parder{}{s^0} + g_2^1\parder{}{s^1} + g_2^2\parder{}{s^2} + g_2^3\parder{}{s^3}\,,\\
	Y_3 =& \ v_3\parder{}{q} + G_3^0\parder{}{p^0} + G_3^1\parder{}{p^1} + G_3^2\parder{}{p^2} + G_3^3\parder{}{p^3} + g_3^0\parder{}{s^0} + g_3^1\parder{}{s^1} + g_3^2\parder{}{s^2} + g_3^3\parder{}{s^3}\,,
\end{align*}
where the functions $G_\alpha^\beta,g_\alpha^\beta$ with $(\alpha,\beta)\in\left(\{0,1,2,3\}\times\{0,1,2,3\}\right)\setminus\{(0,0)\}$ are arbitrary.ma

Notice that conditions \eqref{1-kg}, \eqref{2-kg}, \eqref{3-kg}, \eqref{4-kg}, \eqref{5-kg} and \eqref{6-kg} lead to the equation
$$ \left( \square + m^2 - \gamma_0\parder{}{x^0} + \gamma_1\parder{}{x^1} + \gamma_2\parder{}{x^2} + \gamma_3\parder{}{x^3}\right)\phi = 0\,, $$
which represents a ``damped'' Klein--Gordon equation. Obviously, for $\gamma_\mu = 0$, we recover the Klein--Gordon equation \eqref{klein-gordon}. An important particular case is $\gamma_\mu = (-\gamma, 0, 0, 0)$. In this case, we obtain the telegrapher's equation
$$ \square\phi + \gamma \parder{\phi}{x^0} + m^2\phi = 0 $$
as a particular case of the ``damped'' Klein--Gordon equation.

\subsection{Dissipative Maxwell equations and damped electromagnetic waves}

The behaviour of the electromagnetic field in vacuum is described by Maxwell's equations 
\cite[p.\,2]{Jackson}:
\begin{align}
	& \nabla\cdot E = \frac{\rho}{\epsilon_0}\label{maxwell-1}\,,\\
	& \nabla\cdot B = 0\label{maxwell-2}\,,\\
	& \nabla\times E = -\parder{B}{t}\label{maxwell-3}\,,\\
	& \nabla\times B = \mu_0 J + \mu_0\epsilon_0\parder{E}{t}\label{maxwell-4}\,,
\end{align}

It is well known that we can rewrite Maxwell's equations in the Minkowski Space $\mathbb{M}$ equipped with the 
Minkowski metric $g_{\mu\nu}$,
by defining the electromagnetic tensor $F_{\mu\nu}$ given by
$$ 
F_{\mu\nu} = \parder{A_\nu}{x^\mu} - \parder{A_\mu}{x^\nu} = \partial_\mu A_\nu - \partial_\nu A_\mu = A_{\nu,\,\mu} - A_{\mu,\,\nu}\,, 
$$
where 
$A^\mu = \left(\frac{\phi}{c},A_1,A_2,A_3\right)$ 
is the electromagnetic 4-potential. 
We can also define de current 4-vector as $\mathcal{J}^\mu = (c\rho,J)$. 
With these objects, the first pair of Maxwell's equations \eqref{maxwell-1} and \eqref{maxwell-4} are written as
\begin{equation}\label{maxwell-1st-pair}
	\partial_\mu F^{\mu\nu} = \mu_0\mathcal{J}^\mu\,,
\end{equation}
while the second pair of Maxwell's equations \eqref{maxwell-2} and \eqref{maxwell-3} become
\begin{equation}\label{maxwell-2nd-pair}
	\partial_\alpha F_{\mu\nu} + \partial_\mu F_{\nu\alpha} + \partial_\nu F_{\alpha\mu} = 0\,,
\end{equation}
also known as Bianchi identity. Equations \eqref{maxwell-2nd-pair} are a direct consequence of the definition of $F_{\mu\nu}$, while the first pair of Maxwell's equations \eqref{maxwell-1st-pair} can be obtained as the Euler--Lagrange equations for the Lagrangian
\begin{equation*}
	L = -\frac{1}{4\mu_0}F_{\mu\nu}F^{\mu\nu} - A_\mu\mathcal{J}^\mu\,.
\end{equation*}
From now on, we are going to consider Maxwell's equations without charges and currents ($\mathcal{J}^\mu = 0$),
\begin{gather*}
	\partial_\mu F^{\mu\nu} = 0\,.
\end{gather*}

\subsubsection*{Unified formalism}

Now we are going to develop the unified formalism for the Lagrangian with dissipation \cite{GasMar21}
\begin{equation}\label{lag-maxwell-dis}
	\L = -\frac{1}{4\mu_0}F_{\mu\nu}F^{\mu\nu} - \gamma_\alpha s^\alpha\,,
\end{equation}
defined on the manifold $\oplus^4\T \R^4\times \R^4$ equipped with coordinates $(A_\mu,A_{\mu,\,\nu}\,; s^\alpha)$, where $\mu,\nu,\alpha = 0,1,2,3$ and $\gamma_\alpha = (\gamma_0, \bm{\gamma})$ is a constant 4-vector.

We begin by considering the unified bundle
$$ \W = \oplus^4\T\R^4\times_{\R^4}\oplus^4\T^\ast\R^4\times\R^4\,, $$
equipped with natural coordinates $(A_\mu, A_{\mu,\,\nu}, P^{\mu,\,\nu}, s^\alpha)$. We have the coupling function
$$ \C = P^{\mu,\,\nu} A_{\mu,\,\nu}\,, $$
the canonical forms
$$
	\Theta^\alpha = P^{\mu,\,\alpha}\d A_\mu\quad ,\quad
	\Omega^\alpha = -\d\Theta^\alpha = \d A_\mu\wedge\d P^{\mu,\,\alpha}\,,
$$
and the contact forms
\begin{align*}
	\eta^\alpha = \d s^\alpha - P^{\mu,\,\alpha}\d A_\mu\,.
\end{align*}
Using the Lagrangian \eqref{lag-maxwell-dis}, we define the Hamiltonian function
$$ \H = \C - \L = P^{\mu,\,\nu} A_{\mu,\,\nu} + \frac{1}{4\mu_0}F_{\mu\nu}F^{\mu\nu} + \gamma_\alpha s^\alpha\,. $$
It is easy to check that the vector fields $\Reeb_\alpha = \dparder{}{s^\alpha}$ are Reeb vector fields of $\W$. 
To solve the Lagrangian--Hamiltonian problem for the 4-precontact system $(\W,\eta^\alpha, \H)$ means to find a 4-vector field $\mathbf{Z} = (Z_0,Z_1,Z_2,Z_3)\in\X^4(\W)$ satisfying equations \eqref{k-contact-fields-SR}. We have that
$$ \d\H - \Reeb_\alpha(\H)\eta^\alpha = \left( P^{\mu,\,\nu} - \frac{1}{\mu_0}F^{\mu\nu} \right)\d A_{\mu,\,\nu} + A_{\mu,\,\nu}\d P^{\mu,\,\nu} - \gamma_\alpha P^{\mu,\,\alpha}\d A_\mu\,. $$
Then, consider a 4-vector field $\mathbf{Z} = (Z_0,Z_1,Z_2,Z_3)$ in $\W$ with local expression
$$ Z_\alpha = (Z_\alpha)_\mu\parder{}{A_\mu} + (Z_\alpha)_{\mu\beta}\parder{}{A_{\mu,\,\beta}} + (Z_\alpha)^{\mu\beta}\parder{}{P^{\mu,\,\beta}} + (Z_\alpha)^\beta\parder{}{s^\beta}\,. $$
For this vector field, we have
\begin{align*}
	i(Z_\alpha)\d\eta^\alpha &= (Z_\alpha)_\mu\d P^{\mu,\,\alpha} - (Z_\alpha)^{\mu,\,\alpha}\d A_\mu\,,\\
	i(Z_\alpha)\eta^\alpha &= (Z_\alpha)^\alpha - P^{\mu,\,\alpha}(Z_\alpha)_\mu\,,
\end{align*}
and thus the first equation in \eqref{k-contact-fields-SR} gives the conditions
\begin{align}
	(Z_\alpha)^{\mu\alpha} &= -\gamma_\alpha P^{\mu,\,\alpha} & & \mbox{(coefficients in $\d A_\mu$)}\,,\label{max-cond-1}\\
	P^{\mu,\,\nu} &= \frac{1}{\mu_0}F^{\mu\nu} & & \mbox{(coefficients in $\d A_{\mu,\,\nu}$)}\,,\label{max-cond-2}\\
	A_{\mu,\,\alpha} &= (Z_\alpha)_\mu & & \mbox{(coefficients in $\d P_{\mu,\,\alpha}$)}\,.\label{max-cond-3}
\end{align}
Furthermore, the second equation in \eqref{k-contact-fields-SR} gives
$$ (Z_\alpha)^\alpha = P^{\mu,\,\alpha}\left( (Z_\alpha)_\mu - A_{\mu,\,\alpha} \right) + \L\,, $$
and hence, using \eqref{max-cond-3},
$$ (Z_\alpha)^\alpha = \L\,. $$
We have obtained the constraint functions
$$ \xi^{\mu\nu} = P^{\mu,\,\nu} - \frac{1}{\mu_0}F^{\mu\nu}\,, $$
defining a submanifold $\W_1\hookrightarrow\W$. Now we have to impose the tangecy of the 4-vector field $\mathbf{Z}$ to this submanifold $\W_1$:
\begin{align*}
	0 &= Z_\alpha(\xi^{\mu\nu})
	= Z_\alpha\left(P^{\mu,\,\nu} - \frac{1}{\mu_0}F^{\mu\nu}\right)
	= (Z_\alpha)^{\mu\nu} - \frac{1}{\mu_0}\parder{F^{\mu\nu}}{A_{\tau\beta}}(Z_\alpha)_{\tau\beta}\\
	&= (Z_\alpha)^{\mu\nu} - \frac{1}{\mu_0}\left( g^{\mu\tau}g^{\nu\beta} - g^{\mu\beta}g^{\nu\tau} \right)(Z_\alpha)_{\tau\beta}
\end{align*}
which partially determine some of the coefficients of the 4-vector field $\mathbf{Z}$. Notice that no new constraints appear and hence the constraint algorithm ends with the submanifold $\W_f = \W_1$ and gives the solutions $\mathbf{Z} = (Z_0,Z_1,Z_2,Z_3)$, where
\begin{align*}
	Z_\alpha &= A_{\mu,\,\alpha}\parder{}{A_\mu} + (Z_\alpha)_{\mu\nu}\parder{}{A_{\mu,\,\nu}} + (Z_\alpha)^{\mu\nu}\parder{}{P^{\mu,\,\nu}} + (Z_\alpha)^\beta\parder{}{s^\beta}\,,
\end{align*}
satisfying the conditions
\begin{equation*}
	\begin{dcases}
		(Z_\alpha)^\alpha = \L\,,\\
		(Z_\alpha)^{\mu\alpha} = -\gamma_\alpha P^{\mu,\,\alpha}\,,\\
		(Z_\alpha)^{\mu\nu} = \frac{1}{\mu_0}\left( g^{\mu\tau}g^{\nu\beta} - g^{\mu\beta}g^{\nu\tau} \right)(Z_\alpha)_{\tau\beta}\,.
	\end{dcases}
\end{equation*}

\subsubsection*{4-contact Maxwell equations and damped electromagnetic waves}

Notice that, combining equations \eqref{max-cond-1} and \eqref{max-cond-2}, we obtain
$$ \partial_\alpha F^{\alpha\mu} = -\gamma_\alpha F^{\alpha\mu}\,, $$
which is the dissipative version of the first pair of Maxwell's equations. Together with the Bianchi identity \eqref{maxwell-2nd-pair}, we can write the 4-contact Maxwell's equations without charges and currents:
\begin{align}
	& \nabla\cdot E = -\bm{\gamma}\cdot E\,\label{maxwell-damped-1}\\
	& \nabla\cdot B = 0\,\label{maxwell-damped-2}\\
	& \nabla\times E = -\parder{B}{t}\,\label{maxwell-damped-3}\\
	& \nabla\times B = \mu_0\epsilon_0\parder{E}{t} - \bm{\gamma}\times B + \frac{\gamma_0}{c}E\,.\label{maxwell-damped-4}
\end{align}

Applying the curl operator to the third and fourth equations \eqref{maxwell-damped-3}, \eqref{maxwell-damped-4}, we get
\begin{align*}
	\mu_0\epsilon_0\parder{^2E}{t^2} - \nabla^2E + \frac{\gamma_0}{c}\parder{E}{t} = -\bm{\gamma}\times(\nabla\times E)\,,\\
	\mu_0\epsilon_0\parder{^2B}{t^2} - \nabla^2B + \frac{\gamma_0}{c}\parder{B}{t} = -\nabla\times(\bm{\gamma}\times B)\,.
\end{align*}
Taking $\gamma_\mu = (\gamma_0,\bm{0})$, we obtain
\begin{align*}
	\parder{^2E}{t^2} - c^2\nabla^2E + c\gamma_0\parder{E}{t} = 0\,,\\
	\parder{^2B}{t^2} - c^2\nabla^2B + c\gamma_0\parder{B}{t} = 0\,,
\end{align*}
which are the 3-dimensional analogues of the damped wave equation \eqref{damped-wave} studied in example \ref{exmpl-wave}.

\section{Conclusions and outlook}

We have developed the Skinner--Rusk or unified formalism
for classical field theories with dissipation.
For this, we have started from the geometrical 
Lagrangian and Hamiltonian $k$-contact formalisms previously
introduced \cite{GGMRR-2019,GGMRR-2020} as a generalization of the corresponding 
Lagrangian and Hamiltonian formalisms in contact mechanics
\cite{LGLMR-2021,LGMMR-2020}, and from
the unified formalism for contact mechanics \cite{LGMMR-2020}
and for the $k$-symplectic formulation of classical field theories
\cite{Rey2004}.

The Skinner--Rusk formalism takes place in the so-called Pontryagin bundle
${\cal W}=\oplus^k\T Q\times_Q\oplus^k\T^\ast Q\times\R^k$.
This formalism allows to work comfortably with the field equations,
which are stated in ${\cal W}$,
especially in the case of singular systems.
In particular, the second-order or holonomy condition is
incorporated in a natural way to the solutions
to the equations.
In any case, these equations are not consistent and
the Legendre map is obtained as a first consequence
of the constraint algorithm (from the consistency conditions).
If the Lagrangian describing the system is regular,
the tangency condition in the algorithm leads to the Euler--Lagrange equations and,
using the Legendre map, the Hamilton--de Donder--Weyl equations are obtained.
In this case, no more constraints appear and the final constraint submanifold
is the graph of the Legendre map.
In the singular case, new constraints defining new submanifolds 
can arise as a consequence of the tangency condition.
Once the final constraint submanifold is achieved (when it exists)
the Lagrangian and the Hamiltonian formalisms
(including the field equations, their solutions, and the
constraint submanifolds obtained in the corresponding constraint
algorithms) are obtained by projecting the results
of the Skinner--Rusk formalism in the Pontryagin bundle onto the bundles 
$\oplus^k\T Q\times\R^k$ and $\oplus^k\T^\ast Q\times\R^k$.

We have analyzed three examples.
In all of them we have modified the standard Lagrangians of each system (without dissipation)
by adding a linear term on the extra coordinates of $\R^k$ (the ``dissipation variables'').
The first one is a well-known case, the $1$-dimensional wave equation (vibrating string)
with damping, for which, from the contact field equations, we obtain the classical equation
of this system.
The second one is a very interesting example since, 
after modifying the Klein--Gordon Lagrangian
with an appropriate damping term,
we obtain the telegrapher's equation.
Finally, in the third example, we have modified the 
classical Maxwell Lagrangian in vacuum (without charges and currents)
and, as a final result, we have obtained the equation of electromagnetic waves
with a dissipation term which is similar
to the one in the damped vibrating string equation.

There are some other examples where our formalism could be applied.
In particular, 
it could be interesting to modify the classical Lagrangian of general relativity in the Einstein--Palatini approach
and find physical consequences of the
modified Einstein's equations so obtained.

\subsection*{Acknowledgments}

We acknowledge the financial support from the Spanish Ministerio de Ciencia, Innovaci\'on y Universidades project PGC2018-098265-B-C33 and the Secretary of University and Research of the Ministry of Business and Knowledge of the Catalan Government project 2017--SGR--932.

\bibliographystyle{abbrv}

\end{document}